\documentclass[12pt, draftclsnofoot, onecolumn]{IEEEtran}

\usepackage{amsmath,amssymb,amsthm}
\usepackage{mathtools}
\usepackage{enumitem}
\usepackage{comment}
\usepackage{xcolor}
\usepackage{graphicx}
\usepackage{cite}
\usepackage{tikz}

\definecolor{redish}{rgb}{0.9, 0.17, 0.31}
\definecolor{fuchs}{rgb}{0.57, 0.36, 0.51}
\usepackage[colorlinks=true,linkcolor=fuchs]{hyperref}
\usepackage{cleveref} 

\newtheorem{theorem}{Theorem}
\newtheorem{lemma}[theorem]{Lemma}
\newtheorem{proposition}[theorem]{Proposition}
\newtheorem{corollary}[theorem]{Corollary}
\newtheorem{definition}{Definition}
\newtheorem{example}{Example}
\newtheorem{remark}{Remark}
\newtheorem{problem}{Problem}

\newtheorem{conjecture}{Conjecture}

\newcommand{\E}{\mathbb{E}}
\newcommand{\F}{\mathbb{F}}
\newcommand{\cC}{\mathcal{C}}

\newcommand{\Lam}{\Lambda}
\newcommand{\harm}[1]{H_{#1}}

\newcommand{\taufile}[1]{\tau_{#1}}

\begin{document}

\title{Coded Information Retrieval\\ for Block-Structured DNA-Based Data Storage\vspace{0.5cm}}

\author{Daniella Bar-Lev \\
\textit{\normalsize{Department of Mathematics, Universität Zürich, Switzerland}} \\
\normalsize{email: daniella.bar-lev@math.uzh.ch}
\thanks{This work was supported in part by the National Science Foundation (NSF) under Grant CCF-2212437, by the Swiss National Science Foundation under grant number 212865, and by Schmidt Sciences.}
}

\maketitle

\begin{abstract}
We study the problem of coded information retrieval for block-structured data, motivated by DNA-based storage systems where a database is partitioned into multiple files that must each be recoverable as an atomic unit. We initiate and formalize the block-structured retrieval problem, wherein $k$ information symbols are partitioned into two files $F_1$ and $F_2$ of sizes $s_1$ and $s_2 = k - s_1$. The objective is to characterize the set of achievable expected retrieval time pairs $\bigl(E_1(G), E_2(G)\bigr)$ over all $[n,k]$ linear codes with generator matrix $G$. 
We derive a family of linear lower bounds via mutual exclusivity of recovery sets, and develop a nonlinear geometric bound via column projection  that holds for every linear code. For codes with no mixed columns, this yields the hyperbolic constraint $s_1/E_1 + s_2/E_2 \le 1$, which we conjecture to hold universally whenever $\max\{s_1,s_2\} \ge 2$. We analyze explicit codes, such as the identity code, file-dedicated MDS codes, and the systematic global MDS code, and compute their exact expected retrieval times. For file-dedicated codes we prove MDS optimality within the family and verify the hyperbolic constraint. For global MDS codes, we establish dominance by the proportional local MDS allocation via a convex-ordering argument for hypergeometric distributions, simplifying and extending prior work to the asymmetric case. Finally, we characterize the limiting achievability region as $n \to \infty$: the hyperbolic boundary is asymptotically achieved by file-dedicated MDS codes, and is conjectured to be the exact boundary of the limiting achievability region.
\end{abstract}

\section{Introduction}\label{sec:intro}

\IEEEPARstart{D}{NA}-based data storage has emerged as a compelling medium for long-term archival, offering information density orders of magnitude beyond conventional magnetic and solid-state drives, together with exceptional longevity and negligible energy cost during storage~\cite{church2012next, 
goldman2013towards, grass2015robust}. Pioneering experiments 
demonstrated larger datasets, enhanced robustness via advanced ECC, and capabilities such as random access using polymerase chain reaction (PCR), further validating and extending the pipeline~\cite{bornholt2016dna, blawat2016forward, erlich2017dna, 
yazdi2017portable, organick2018random, anavy2019data,
chandak2020overcoming, bar2025scalable}.
Collectively, these milestones validate the DNA-based data storage pipeline, while highlighting the emerging computational and architectural challenges that must be addressed to scale these systems practically \cite{bar2024zettabyte}.

In a conventional DNA-based data storage pipeline, digital information is translated into sequences of biological nucleotides and chemically synthesized into DNA molecules. Due to biochemical challenges, data cannot be written as a single continuous string; instead, it is partitioned into millions of short strands (typically 250 to 300 nucleotides). As these strands are preserved together in a structureless, unordered biochemical pool, each must be augmented with an index to encode its original position within the data. During synthesis, multiple redundant noisy copies of each strand are generated. Data retrieval is performed via high-throughput sequencing, which samples molecules from this pool to produce a large multiset of reads affected by unique channel errors such as insertions, deletions, and substitutions. Recovering the original digital files from this unordered, error-prone output relies heavily on specialized error-correcting codes~\cite{doricchi2022emerging, shomorony2022information, sabary2024survey, milenkovic2024dna}.

Because the sequencing process draws strands uniformly at random from the pool, recovering the data inherently requires oversampling. Since sequencing remains a primary bottleneck of the pipeline in terms of both cost and latency, a critical performance metric for any DNA-based storage system is its \emph{coverage depth}, defined as the ratio of the total number of sequenced reads to the number of uniquely designed strands.

The mathematical structure underlying this coverage depth requirement is strongly connected to the classical \emph{Coupon Collector's Problem} (CCP) which can be formalized as follows. Given $n$ distinct coupon types drawn uniformly at random with replacement, what is the expected number of draws until every type has been seen at least once? The probabilistic theory of the CCP (see e.g. \cite{erdHos1961classical, newman1960double, flajolet1992birthday}) provides the foundational framework for analyzing such random-access retrieval systems.

The theoretical study of coverage depth for DNA storage was initiated in~\cite{bar-lev2025cover}, which formalized both the \emph{full recovery} and \emph{random-access} coverage depth problems. For full recovery, it was established that Maximum Distance Separable (MDS) codes are optimal for minimizing the expected retrieval time. However, for random access, where a user seeks to retrieve a single specific information symbol rather than decoding the entire database, the study revealed a surprising limitation: systematic MDS codes require an expected retrieval time of at least $k$, offering no improvement over an uncoded baseline. To overcome this, new coding schemes were proposed that successfully achieve expected retrieval times strictly below $k$.

Building on this foundational work, Gruica et al.~\cite{gruica2024combinatorial} introduced new combinatorial techniques to capture the structural properties that enable these random-access improvements. By leveraging column dependencies and recovery set intersections, they derived exact closed-form expressions for the expected retrieval time of arbitrary linear codes via subset counting. A central contribution of their work was the identification of \emph{recovery-balanced codes}---a geometric property satisfied by classical families such as MDS, Simplex, and Hamming codes, which strictly bounds their expected retrieval time to $k$. Crucially, they demonstrated that deliberately breaking this balance enables retrieval times strictly less than $k$, providing systematic methods for designing highly efficient random-access codes. The exact subset-counting framework from~\cite{gruica2024combinatorial} serves as a central technical tool in the present paper.

The geometric structure of codes that minimize retrieval time was subsequently analyzed in~\cite{gruica2024geometry}. This theoretical pursuit has expanded rapidly. Boruchovsky et al.~\cite{boruchovsky2025making} determined optimal constructions for specific small values of $k$, while Bodur et al.~\cite{bodur2025random} analyzed the full probability distribution of retrieval times. The fundamental limits of this problem have been further refined for small alphabets through the lens of weight distributions and duality~\cite{bertuzzo2026dna}. Concurrently, linear-time algorithms for computing exact retrieval expectations, alongside improved theoretical bounds and optimal code constructions for small parameters, were developed in~\cite{wang2026random}. Shifting to the full-recovery problem, Hanna~\cite{hanna2025reliability} investigated the reliability of MDS-coded retrieval over noisy channels, analyzing the effects of redundancy allocation within both the inner and the outer codes.

Beyond the assumption of uniform random access over a centralized pool, the coverage depth problem has recently been extended to capture more complex sequencing dynamics and system architectures. Cao and Chen~\cite{cao2025optimizing} generalized the analysis using insights from experimental DNA storage data to account for non-uniform sampling. Similarly, Grunbaum and Yaakobi~\cite{grunbaum2025general} expanded the mathematical framework to contiguous window sampling models motivated by shotgun sequencing, while Levy et al.~\cite{levy2026expected} analyzed expected recovery times in the context of DNA-based distributed storage systems. Finally, this probabilistic perspective has proven versatile enough to optimize coverage ratios in entirely different DNA-based data storage architectures, such as the combinatorial motif-based model, which encodes information using subsets of pre-synthesized shortmer sequences~\cite{preuss2024sequencing, sokolovskii2024coding}, and the composite DNA paradigm, which leverages statistical mixtures of nucleotides at each sequence position~\cite{cohen2025optimizing}.

While the literature has extensively explored the two extremes of data retrieval---accessing a single isolated strand or recovering the entire database---practical storage systems typically operate between these bounds. Real-world user data is rarely accessed as an isolated symbol. Instead, data is organized into \emph{files}, that is, atomic blocks of strands that must be recovered in their entirety. A digital image or compressed archive corresponds to a specific subspace of information strands that remains unusable until all its constituents are successfully decoded. This operational reality necessitates a shift to a \emph{block-structured retrieval} model. A coding strategy optimized for single-strand retrieval does not inherently optimize file recovery; the recovery time of a file depends critically on its dimension and on how the code distributes redundancy across complementary files sharing the storage medium. Recently, Abraham et al.~\cite{abraham2024covering} investigated this multi-file setting using local MDS codes for symmetric file partitions, demonstrating improvement over global systematic MDS codes. However, the general asymmetric case, where files are of arbitrary sizes, was left entirely uncharacterized, and whether local MDS codes are optimal among all linear codes remained open.

In this paper, we establish the theoretical foundations of block-structured coded information retrieval for DNA-based data storage. We formalize the problem wherein $k$ information symbols are partitioned into two files $F_1$ and $F_2$ of arbitrary sizes $s_1$ and $s_2 = k - s_1$, and define the \emph{achievability region} as the set of expected retrieval time pairs $(E_1(G), E_2(G))$ attainable by any $[n,k]$ linear code with generator matrix $G$. Our contributions are 
as follows.

\begin{enumerate}[label=(\roman*)]
\item We first derive a family of \emph{linear} lower bounds via mutual exclusivity of recovery sets (Section~\ref{sec:lower_bounds}). While these already carve out an infeasible region for all finite $n,k$, we show they cannot capture the true trade-off, which is \emph{hyperbolic} in nature --- motivating the nonlinear analysis
that follows. The reader interested primarily in the hyperbolic bounds may proceed directly to Section~\ref{sec:nonlinear_bound}.

\item We develop a nonlinear geometric bound via column projection that holds for every linear code (Theorem~\ref{thm:nonlinear_bound}), which for codes with no mixed columns yields the hyperbolic constraint $s_1/E_1 + s_2/E_2 \le 1$ (Corollary~\ref{cor:pure_file_hyperbola}). We conjecture this holds universally for all linear codes when $\max\{s_1,s_2\} \ge 2$ and provide supporting evidence including verification for all code families studied in this paper.

\item We compute exact expected retrieval times for three code families: the identity code (Section~\ref{ssec:identity}), file-dedicated MDS codes (Section~\ref{ssec:dedicated}), and systematic global MDS codes (Section~\ref{ssec:MDS_construction}). For file-dedicated codes, we prove that MDS is optimal within the family. For global MDS codes, we prove dominance by the proportional local MDS allocation (Proposition~\ref{prop:local_dominates_global}) by expressing both expected retrieval times as the mean of a common convex function of a hypergeometric random variable, providing a significantly simpler proof of results from~\cite{abraham2024covering} and extending them to asymmetric partitions when $k \mid n$.

\item We characterize the limiting achievability region as ${n \to \infty}$, and show that the full hyperbolic boundary $\{s_1/E_1 + s_2/E_2 = 1, E_i \ge s_i\}$ is asymptotically achieved by file-dedicated MDS codes (Theorem~\ref{thm:asymptotic_ub}), and is conjectured to be the exact boundary of the limiting achievability region (Conjecture~\ref{conj:limiting}).
\end{enumerate}

The remainder of this paper is organized as follows. Section~\ref{sec:system_model} details the system model and geometric preliminaries. Section~\ref{sec:motivation} formulates the block-structured retrieval problem. Section~\ref{sec:lower_bounds} derives the theoretical linear lower bounds. Section~\ref{sec:nonlinear_bound} develops the nonlinear bound and states the universal conjecture. Section~\ref{sec:constructions} analyzes explicit code constructions and evaluates their expected retrieval times. Section~\ref{sec:asymptotics} analyzes the limiting achievability region. Section~\ref{sec:conclusion} concludes the paper and outlines open problems.
\section{Definitions and Preliminaries}\label{sec:system_model}

Throughout this paper, $k$ and $n$ are positive integers with $k \le n$, $q$ is a prime power, $\F_q$ is the finite field with $q$ elements, and $[m] \coloneqq \{1, \ldots, m\}$. The $r$-th harmonic number is denoted by $\harm{r} \coloneqq \sum_{i=1}^{r} \frac{1}{i}$, with $\harm{0} = 0$.

\subsection{Encoding Model and File Structure}

We consider a storage system where $k$ information symbols $\mathbf{u} = (u_1, \ldots, u_k) \in \F_q^k$ are encoded into $n$ encoded symbols $(x_1, \ldots, x_n) = \mathbf{u} G$ using a generator matrix $G \in \F_q^{k \times n}$ of rank $k$. Let $g_j$ denote the $j$-th column of $G$. 

During retrieval, encoded symbols are drawn uniformly at random with replacement. Therefore, at each step, any column $g_j$ is drawn with probability $1/n$.

\begin{definition}[Recovery Set \cite{gruica2024combinatorial}]\label{def:recovery_set}
A subset of drawn indices $S \subseteq [n]$ forms a \emph{recovery set} for an information symbol $u_i$ if the standard basis vector $e_i$ lies within the span of the corresponding drawn columns, i.e., $e_i \in \langle g_j : j \in S \rangle_{\F_q}$.
\end{definition}

For simplicity, whenever the specific field $\F_q$ is clear, we will omit it from the notations. 

\begin{definition}[Subset Count \cite{gruica2024combinatorial}]\label{def:alpha_subset}
For any target subset of information symbols indexed by $I \subseteq [k]$, and for $0 \le s \le n$, the subset count $\alpha_I(s)$ is the number of $s$-subsets of $[n]$ that successfully recover all symbols in $I$:
\[
  \alpha_I(s) \;\coloneqq\; \bigl|\{S \subseteq [n] : |S| = s,\; \{e_i : i \in I\} \subseteq \langle g_j : j \in S \rangle\}\bigr|.
\]
\end{definition}

\begin{lemma}[{\cite[Remark 2]{gruica2024combinatorial}}]\label{lem:exp_formula_subset}
For any generator matrix $G \in \F_q^{k \times n}$ and target index set $I \subseteq [k]$, the expected number of uniform draws with replacement to recover all symbols in $I$ is given by:
\begin{equation}\label{eq:exp_formula_subset}
  \E[\tau_I(G)] \;=\; n\harm{n} \;-\; \sum_{s=1}^{n-1} \frac{\alpha_I(s)}{\binom{n-1}{s}}.
\end{equation}
\end{lemma}

Unlike single-strand retrieval, which was considered in e.g.~\cite{bar-lev2025cover, gruica2024geometry, boruchovsky2025making}, here we assume user data is organized into block-structured files. Let the set of symbols be partitioned into $f \ge 2$ disjoint files $F_1, \ldots, F_f$,
where $F_i$ has size $|F_i| = s_i \ge 1$ and $\sum_{i=1}^{f} s_i = k$. Each file must
be recovered as an atomic unit.

In this work we focus on the two-file case $f = 2$, which already exhibits the
essential trade-off between the retrieval times of complementary files; the general
multi-file case is discussed as an open direction in Section~\ref{sec:conclusion}.
For $f = 2$ we write
\[
  F_1 = \{u_1, \ldots, u_{s_1}\}, \qquad F_2 = \{u_{s_1+1}, \ldots, u_k\},
\]
with $|F_1| = s_1 \ge 1$, $|F_2| = s_2 \ge 1$, and $s_1 + s_2 = k$.

For algebraic analysis, recovering the symbols in a file $F_i$ is mathematically equivalent to the drawn columns spanning the $s_i$-dimensional subspace defined by the corresponding standard basis vectors: $\langle e_j : u_j \in F_i \rangle_{\F_q}$. We will interchangeably use $F_i$ to denote this target subspace.

\begin{definition}[File Retrieval Time]\label{def:file_tau}
For a generator matrix $G \in \F_q^{k \times n}$ and $i \in \{1,2\}$, let $\taufile{F_i}(G)$ denote the minimum number of random draws needed to recover \emph{all} information symbols in file $F_i$. We denote its expected value by $E_i(G) \coloneqq \E[\taufile{F_i}(G)]$.
\end{definition}

By substituting the target file into Lemma~\ref{lem:exp_formula_subset}, the expected retrieval time of file $F_i$ can be evaluated exactly as:
\begin{equation}\label{eq:file_exp}
  E_i(G) \;=\; n\harm{n} \;-\; \sum_{s=1}^{n-1} \frac{\alpha_{F_i}(s)}{\binom{n-1}{s}}.
\end{equation}

\subsection{Subspace Geometries and Column Counting}

To evaluate file-level recovery bounds, we must track the progression of the drawn columns as a growing geometric subspace. Let $\mathcal{L}(G)$ denote the \emph{subspace lattice} generated by $G$, defined as the finite collection of all subspaces spanned by subsets of columns of $G$, partially ordered by inclusion.

\begin{definition}[Column Counting Function]\label{def:column_count}
For any subspace $W \subseteq \F_q^k$, we define $N(W)$ as the exact number of columns of $G$ that are entirely contained within $W$:
\[
  N(W) \;\coloneqq\; \bigl|\{j \in [n] : g_j \in W\}\bigr|.
\]
\end{definition}

\begin{lemma}\label{lem:column_inc_exc}
For any two subspaces $A, B \subseteq \F_q^k$, the column counting function satisfies:
\[
  N(A) + N(B) \;=\; N(A \cup B) + N(A \cap B),
\]
where $A \cup B$ denotes the set-theoretic union of the two subspaces.
\end{lemma}
\begin{IEEEproof}
Let $C_G$ be the multi-set of columns of $G$. The function $N(W)$ counts the cardinality of the intersection $|C_G \cap W|$. By the fundamental inclusion-exclusion principle for finite sets, $|C_G \cap A| + |C_G \cap B| = |(C_G \cap A) \cup (C_G \cap B)| + |(C_G \cap A) \cap (C_G \cap B)|$. Since intersection distributes over union, the right hand side is exactly $|C_G \cap (A \cup B)| + |C_G \cap (A \cap B)|$, which directly yields $N(A \cup B) + N(A \cap B)$.
\end{IEEEproof}

\section{Motivation and Problem Formulation}\label{sec:motivation}

A natural baseline for block-structured retrieval is the \emph{local MDS} strategy: partition the $n$ available encoded strands into two dedicated groups of sizes $n_1$ and $n_2 = n - n_1$, with $n_i \geq s_i$, and protect each file independently with an MDS code. Since each draw from the pool of $n$ strands hits file $F_i$'s dedicated allocation with probability $n_i/n$, the expected retrieval time reduces to a generalized coupon-collector problem:
\begin{equation}\label{eq:local_mds}
    E_i\!\left(G_\mathrm{Local}^{(n_1,n_2)}\right) 
    \;=\; \sum_{j=1}^{s_i} \frac{n}{n_i - j + 1}.
\end{equation}
For fixed $n$, the integer constraint $n_1 + n_2 = n$ (with $n_i \geq s_i$) limits the local strategy to finitely many valid operating points, each reflecting a strict design trade-off between the retrieval times of the two files.

Prior work~\cite{abraham2024covering} compared the uniform local MDS code against a global systematic MDS code for the symmetric case $s_1 = s_2$, establishing that the former achieves no worse expected retrieval time for either file.

\begin{lemma}[{\cite[Lemma~2]{abraham2024covering}}]\label{lem:abraham_symmetric}
For $s_1 = s_2$ and $n_1 = n_2 = n/2$, the uniform local and global systematic MDS codes satisfy
$E_i(G_\mathrm{Global}) \ge E_i(G_\mathrm{Local})$ for $i \in \{1,2\}$.
\end{lemma}

This result has two important limitations. First, it is restricted to equal file sizes with a uniform redundancy allocation, leaving the asymmetric case entirely open. Second, it only establishes that one specific alternative (global MDS) is suboptimal relative to local MDS; it does not address whether local MDS is itself optimal among all linear codes, for either symmetric or asymmetric files. We investigate both gaps through the following examples.

Since improving $E_1$ generally comes at the cost of $E_2$, comparing two generator matrices requires care. We say that a generator matrix $G'$ \emph{dominates} $G$ if
\begin{equation}\label{eq:pareto}
    E_i(G') \;\le\; E_i(G) \quad \text{for both } i \in \{1,2\},
\end{equation}
with strict inequality for at least one $i$. A operating point that is not dominated by any other operating point is called \emph{Pareto-optimal}. Under this criterion, comparing two generator matrices that each favor a different file --- one achieving lower $E_1$ and the other lower $E_2$ --- admits no clear winner, and both may be Pareto-optimal.
\begin{example}\label{ex:asymmetric_n8_local}
Consider an asymmetric partition with $k=4$, $s_1=1$, $s_2=3$, and $n=8$, providing $n - k = 4$ redundant symbols. Evaluating~\eqref{eq:local_mds} over all valid integer allocations
$n_1 + n_2 = 8$ yields five operating points:
\begin{itemize}
    \item \textbf{Case A:} $n_1=1$ and $n_2=7$. By \eqref{eq:local_mds}, $E_1 = 8/1=8$, and
    \begin{equation*}
        E_2(G_\text{Local}^{(1,3,1,7)}) = \frac{8}{7} + \frac{8}{6} + \frac{8}{5} = \frac{428}{105} \approx 4.08.
    \end{equation*}
    \item \textbf{Case B:} $n_1=2$ and $n_2=6$. By \eqref{eq:local_mds}, $E_1 = 8/2=4$, and
    \begin{equation*}
        E_2(G_\text{Local}^{(1,3,2,6)}) = \frac{8}{6} + \frac{8}{5} + \frac{8}{4} = \frac{74}{15} \approx 4.93.
    \end{equation*}
    \item \textbf{Case C:} $n_1=3$ and $n_2=5$. By \eqref{eq:local_mds}, $E_1 = 8/3 \approx 2.67$, and
    \begin{equation*}
        E_2(G_\text{Local}^{(1,3,3,5)}) = \frac{8}{5} + \frac{8}{4} + \frac{8}{3} = \frac{94}{15} \approx 6.27.
    \end{equation*}
    \item \textbf{Case D:} $n_1=4$ and $n_2=4$. By \eqref{eq:local_mds}, $E_1 = 8/4=2$, and
    \begin{equation*}
        E_2(G_\text{Local}^{(1,3,4,4)}) = \frac{8}{4} + \frac{8}{3} + \frac{8}{2} = \frac{26}{3} \approx 8.67.
    \end{equation*}
    \item \textbf{Case E:} $n_1=5$ and $n_2=3$. By \eqref{eq:local_mds}, $E_1 = 8/5=1.6$, and
    \begin{equation*}
        E_2(G_\text{Local}^{(1,3,5,3)}) = \frac{8}{3} + \frac{8}{2} + \frac{8}{1} = \frac{44}{3} \approx 14.67.
    \end{equation*}
\end{itemize}
Moving from Case~A to Case~E, $E_1$ decreases monotonically while $E_2$ increases, and no case dominates any other in the sense of~\eqref{eq:pareto}. For comparison, the systematic $[8,4]$ MDS code with generator matrix $G_\mathrm{Global}$ achieves $E_1(G_\mathrm{Global}) = 4$ and, evaluating~\eqref{eq:file_exp} with subset counts $\alpha_{F_2}(s) = 0,0,1,70,56,28,8$ for $s=1,\ldots,7$, respectively:
\begin{align*}
    E_2(G_\mathrm{Global})
    &= 8H_8 - \!\left(
        \frac{0}{7}\hspace{-0.3ex}+\hspace{-0.3ex}\frac{0}{21}\hspace{-0.3ex}+\hspace{-0.3ex}\frac{1}{35}\hspace{-0.3ex}+\hspace{-0.3ex}\frac{70}{35}\hspace{-0.3ex}+\hspace{-0.3ex}\frac{56}{21}\hspace{-0.3ex}+\hspace{-0.3ex}\frac{28}{7}\hspace{-0.3ex}+\hspace{-0.3ex}\frac{8}{1}\right)\\
     &= \frac{106}{21} \approx 5.05.
\end{align*}
At the same value of $E_1 = 4$, global MDS achieves $E_2 \approx 5.05$,
which is strictly worse than Case~B's $E_2 \approx 4.93$. Global MDS
is therefore dominated by Case~B.
\end{example}

Example~\ref{ex:asymmetric_n8_local} suggests that local MDS may already be hard to improve upon. However, the five operating points of  Example~\ref{ex:asymmetric_n8_local} are the \emph{only} ones the local  strategy can reach: a designer requiring a value of $E_1$ strictly  between two consecutive cases has no recourse within this family. The question of whether the gaps between these discrete points are achievable by other codes --- and if so, at what cost to $E_2$ --- is entirely open from the local MDS perspective alone.

\begin{example}\label{ex:hybrid_n8}
With the same parameters ($k=4$, $s_1=1$, $s_2=3$, $n=8$), consider the code $\mathcal{C}^*$ with encoded strands
\[
\{u_1,\ u_2,\ u_3,\ u_4,\ u_1+u_2,\ u_2+u_3,\ u_3+u_4,\ u_4+u_1\}.
\]
The generator matrix is therefore
\[
G^* = \begin{pmatrix}
1 & 0 & 0 & 0 & 1 & 0 & 0 & 1 \\
0 & 1 & 0 & 0 & 1 & 1 & 0 & 0 \\
0 & 0 & 1 & 0 & 0 & 1 & 1 & 0 \\
0 & 0 & 0 & 1 & 0 & 0 & 1 & 1
\end{pmatrix}.
\]
As shown in~\cite{bar-lev2025cover}, $E_1(\mathcal{C}^*) = \frac{403}{105} \approx 3.84$.

We now compute $E_2(\mathcal{C}^*)$ for $F_2 = \{u_2,u_3,u_4\} = \langle e_2,e_3,e_4 \rangle$. The key structural observation is that exactly five columns of $G^*$ lie in $F_2$ (those with zero $e_1$-component): $\{g_2,g_3,g_4,g_6,g_7\}$, and three lie outside: $\{g_1,g_5,g_8\}$.

We compute $\alpha_{F_2}(s)$ for each $s$:
\begin{itemize}
\item $s=1,2$: The span of at most 2 columns cannot contain the 3-dimensional space $F_2$, so $\alpha_{F_2}(1)=\alpha_{F_2}(2)=0$.

\item $s=3$: The span must equal $F_2$ exactly, requiring all 3 columns to lie in $F_2$. Of the $\binom{5}{3}=10$ three-subsets of $\{g_2,g_3,g_4,g_6,g_7\}$, exactly two are rank-deficient: $\{g_2,g_3,g_6\}$ (since $g_6=g_2+g_3$) and $\{g_3,g_4,g_7\}$ (since $g_7=g_3+g_4$). Thus $\alpha_{F_2}(3)=8$.

\item $s=4$: We count by the number of columns drawn from outside $F_2$:
\begin{itemize}
    \item \emph{0 outside:} All $\binom{5}{4}=5$ four-subsets of     $F_2$-columns span $F_2$, contributing $5$.
    \item \emph{1 outside:} The 8 spanning three-subsets of $F_2$-columns, each paired with any of 3 outside columns, contribute $24$.
    \item \emph{2 outside:} Each outside pair provides a fixed 2D span with an $e_1$-component: $\{g_1,g_5\}$ yields $\langle e_1,e_2\rangle$, $\{g_1,g_8\}$ yields $\langle e_1,e_4\rangle$, and $\{g_5,g_8\}$ yields $\langle e_1+e_2, e_1+e_4\rangle$. Checking all 10 inside pairs for each case, we find 5, 5, and 8 valid pairs respectively, contributing $18$.
    \item \emph{3 outside:} The three outside columns together span $\langle e_1,e_2,e_4\rangle$. An inside column recovers $F_2$ iff it introduces $e_3$, which holds for $g_3,g_6,g_7$ only, contributing $3$.
\end{itemize}
Thus $\alpha_{F_2}(4) = 5+24+18+3 = 50$.

\item $s=5$: A 5-subset fails only if all 5 columns lie in a 3D subspace not containing $F_2$. Each of the three such subspaces contains exactly 5 of our columns: $\langle e_1,e_2,e_3\rangle = \{g_1,g_2,g_3,g_5,g_6\}$, $\langle e_1,e_3,e_4\rangle = \{g_1,g_3,g_4,g_7,g_8\}$, and $\langle e_1,e_2,e_4\rangle = \{g_1,g_2,g_4,g_5,g_8\}$, giving exactly 3 failing subsets. Thus $\alpha_{F_2}(5) = \binom{8}{5}-3 = 53$.

\item $s=6,7$: Since no proper 3D subspace contains more than 5 of our columns, every such subset spans $\mathbb{F}_q^4 \supseteq F_2$, giving $\alpha_{F_2}(6)=28$ and $\alpha_{F_2}(7)=8$.
\end{itemize}

Substituting into \eqref{eq:file_exp}:
\begin{align*}
E_2(\mathcal{C}^*) &= 8H_8 - \left(
    \frac{0}{7} + \frac{0}{21} + \frac{8}{35} + \frac{50}{35} 
    + \frac{53}{21} + \frac{28}{7} + \frac{8}{1}\right) \\
&= \frac{2283}{105} - \frac{1699}{105} 
 = \frac{584}{105} \approx 5.562.
\end{align*}

This point lies strictly between Cases~B and~C in $E_1$, a region entirely inaccessible to the local MDS family. A system requiring $E_1 < 4$ is otherwise forced to Case~C, incurring $E_2 \approx 6.27$.  ${G}^*$ meets the requirement while reducing $E_2$ by more than $11\%$ relative to Case~C. Moreover, while ${G}^*$ achieves lower $E_1$ than Case~B at the cost of higher $E_2$, neither dominates the other, and the choice between them depends entirely on the system's design requirements.
\end{example}

Taken together, these examples reveal that neither local MDS nor any other single code family is sufficient to characterize the full landscape of achievable trade-offs when $n$ is fixed. The set of all Pareto-optimal operating points is precisely the boundary of the \emph{achievability region}, that is, the set of all $(E_1, E_2)$ pairs attainable by any linear code. Every point on this boundary is optimal in the sense of~\eqref{eq:pareto}, and every point strictly outside it is unachievable. Characterizing this boundary for all file sizes, for finite $n$, and in the limit $n \to \infty$, is the central goal of this paper.

\begin{problem}[Coded File Retrieval]\label{prob:main}
Given $n$, $k$, and a partition $(s_1, \ldots, s_f)$ with $\sum_{i=1}^{f} s_i = k$,
the {achievability region} is
\begin{equation}\label{eq:achievability_region}
  \Lambda_{n,k}(s_1, \ldots, s_f) \hspace{-0.5ex}\coloneqq
  \left\{\left(E_1(G),\, \ldots,\, E_f(G)\right)
  : \substack{G \in \mathbb{F}_q^{k \times n},\\ \mathrm{rank}(G) = k}\right\}
  \hspace{-0.5ex}\subseteq \mathbb{R}^f.
\end{equation}
The objectives of this work are:
\begin{enumerate}[label=(\roman*)]
    \item Derive necessary conditions on $(E_1, \ldots, E_f)$ holding for every
    generator matrix, characterizing the non-achievable region.
    \item Construct matrices tracing the Pareto boundary of
    $\Lambda_{n,k}(s_1,\ldots,s_f)$ for finite $n$.
    \item Determine the limiting boundary as $n \to \infty$.
\end{enumerate}
\end{problem}

For the remainder of the paper, we focus on the case where $f = 2$ and write $\Lambda_{n,k}(s_1,s_2)$; see Section~\ref{sec:conclusion} for the multi-file case.

\section{Linear Lower Bounds}\label{sec:lower_bounds}

In this section we derive necessary conditions on $(E_1, E_2)$ that hold for every rank-$k$ generator matrix with parameters $(n,k)$. Each is given in closed form, and together they define a polytope of pairs that no code can achieve. Beyond delimiting the infeasible region directly, these bounds serve as the starting point for Section~\ref{sec:nonlinear_bound}: the polytope of Propositions~\ref{prop:basic_lb} and~\ref{prop:joint_lb} yields the first nonlinear bounds on $s_1/E_1 + s_2/E_2$ (Corollaries~\ref{cor:smax_bound}  and~\ref{cor:cs_bound}), which the column-geometry argument then sharpens toward the hyperbolic curve the operating points trace.


\subsection{Basic lower bounds}

\begin{proposition}\label{prop:basic_lb}
For any $(E_1, E_2) \in \Lam_{n,k}(s_1, s_2)$ and any generator matrix $G \in \F_q^{k \times n}$,
\begin{align}
  E_i(G) &\;\ge\; s_i, \qquad i = 1, 2. \label{eq:lb_each}
\end{align}
\end{proposition}

\begin{IEEEproof}
Let $\mathcal{L}(G)$ denote the lattice of subspaces spanned by subsets of columns of $G$. For any sequence of draws, let $V_t = \langle g_{\xi_1}, \dots, g_{\xi_t} \rangle \in \mathcal{L}(G)$ be the subspace spanned by the drawn columns up to time $t$. 
Because each drawn column $g_{\xi_t}$ is a single vector in $\F_q^k$, the dimension of the spanned subspace can increase by at most $1$ per draw:
\[
  \dim(V_t) \;\le\; \dim(V_{t-1}) + 1.
\]
By induction, starting from the trivial subspace $V_0 = \{0\}$, we have the deterministic geometric bound $\dim(V_t) \le t$ for all $t \ge 0$.

By definition, the stopping time $\tau_i = \taufile{F_i}(G)$ is the minimum time $t$ such that the target file's basis is fully contained within the spanned subspace, meaning $F_i \subseteq V_{\tau_i}$. Hence,
\[
  s_i \;=\; \dim(F_i) \;\le\; \dim(V_{\tau_i}).
\]
Combining these bounds yields $s_i \le \tau_i$ for every possible realization of the random drawing sequence. Taking the expectation over the uniform drawing distribution preserves the inequality, giving $E_i(G) = \E[\tau_i] \ge s_i$.
\end{IEEEproof}

\begin{proposition}\label{prop:joint_lb}
For any $(E_1, E_2) \in \Lam_{n,k}(s_1, s_2)$,
\begin{equation}\label{eq:joint_lb}
  E_1 + E_2 \;\ge\; k + \min(s_1, s_2).
\end{equation}
\end{proposition}

\begin{IEEEproof}
Let $\tau_1 = \taufile{F_1}(G)$ and $\tau_2 = \taufile{F_2}(G)$ be the random variables denoting the number of draws required to recover files $F_1$ and $F_2$ using some rank-$k$ matrix $G
\in\F_q^{k\times n}$, respectively. 
For any realization of the drawing process, the sum satisfies:
\begin{equation}\label{eq:rv_sum}
  \tau_1 + \tau_2 \;=\; \max(\tau_1, \tau_2) + \min(\tau_1, \tau_2).
\end{equation}

We analyze the two terms on the right-hand side individually.
First, because $F_1$ and $F_2$ partition the basis of the $k$-dimensional information space, recovering both files is structurally equivalent to recovering the entire space $\F_q^k$. Therefore, the maximum of the two stopping times is exactly the time required to recover all $k$ information symbols, i.e., 
\[
  \max(\tau_1, \tau_2) \;=\; \taufile{F_1 \cup F_2}(G).
\]
Since each draw increases the dimension of the span by at most $1$, at least $k$ draws are needed to span $\F_q^k$, and thus $\max(\tau_1, \tau_2) \ge k$ for every realization of the drawing process.

Second, the minimum of the two stopping times represents the draws needed to recover the faster of the two files. Because it requires at least $s_i$ draws to span an $s_i$-dimensional subspace, we have
\[
  \min(\tau_1, \tau_2) \;\ge\; \min(s_1, s_2) 
\]
which also holds for every realization of the drawing process.

Substituting these deterministic lower bounds back into \eqref{eq:rv_sum} yields:
\[
  \tau_1 + \tau_2 \;\ge\; k + \min(s_1, s_2) 
\]
for every valid sequence of draws.  Taking expectations preserves the inequality, completing the proof.
\end{IEEEproof}While Proposition~\ref{prop:joint_lb} establishes a universal deterministic limit, we can derive a strictly sharper bound by accounting for the probabilistic cost of drawing from a finite pool of $n$ encoded strands. By bounding the expected time to reach specific dimensions against the optimal performance of MDS codes for the recovery of the entire $k$ symbols~\cite{bar-lev2025cover}, we obtain the following finite-length constraint.

\begin{theorem}\label{thm:mds_rank_sum}
For any rank-$k$ generator matrix $G \in \F_q^{k \times n}$ and any partition $s_1 + s_2 = k$, the achievable expected retrieval times satisfy the  finite-length bound:
\begin{equation}\label{eq:mds_rank_sum}
  E_1(G) + E_2(G) \;\ge\; n\bigl(2H_n - H_{n-k} - H_{n-\min(s_1, s_2)}\bigr).
\end{equation}
\end{theorem}\begin{IEEEproof}
By the identity established in the proof of Proposition~\ref{prop:joint_lb}, taking the expectation over the uniform drawing distribution yields 
\begin{equation}\label{eq:exp_sum_identity}
    E_1(G) + E_2(G) \;=\; \E[\max(\tau_1, \tau_2)] + \E[\min(\tau_1, \tau_2)].
\end{equation}

For the first term, $\max(\tau_1, \tau_2)$ is the time required for the drawn columns of $G$ to reach full rank $k$. Let $V_{(j)} \in \mathcal{L}(G)$ denote the random subspace spanned by the drawn columns at the exact moment its dimension reaches $j$. To transition to dimension $j+1$, the sequence of draws must strictly escape $V_{(j)}$. By Definition~\ref{def:column_count}, exactly $N(V_{(j)})$ columns of $G$ are contained in $V_{(j)}$. Because the draws are uniform with replacement, the expected number of draws to escape $V_{(j)}$ is exactly $\frac{n}{n - N(V_{(j)})}$. 

Since $V_{(j)}$ is spanned by columns of $G$ and has dimension~$j$, it must contain at least $j$ linearly independent columns, meaning $N(V_{(j)}) \ge j$ for any linear code. Therefore, the expected escape time satisfies
\begin{equation}\label{eq:escape_time_bound}
    \frac{n}{n - N(V_{(j)})} \;\ge\; \frac{n}{n - j}.
\end{equation}
Summing these sequential lower bounds from $j=0$ to $k-1$ yields the minimum expected time to reach full rank. This sum is precisely the expected time for an $[n,k]$ MDS code (where $N(V_{(j)}) = j$ strictly holds for all $j < k$). Thus, $\E[\max(\tau_1, \tau_2)] \ge n(H_n - H_{n-k})$.

For the second term, at the random stopping time $t^* = \min(\tau_1, \tau_2)$, the drawn subspace must fully contain either $F_1$ or $F_2$. This enforces that its dimension must be at least $\min(s_1, s_2)$. Consequently, the sequence of drawn subspaces must undergo at least $\min(s_1, s_2)$ strict dimension increments before the faster file can be recovered. Applying the exact same sequence of structural escape time bounds \eqref{eq:escape_time_bound} up to $j = \min(s_1, s_2) - 1$ yields:
\begin{align*}\label{eq:min_tau_bound}
    \E[\min(\tau_1, \tau_2)] \;&\ge\; \sum_{j=0}^{\min(s_1, s_2)-1} \frac{n}{n - j} \;\\&=\; n\bigl(H_n - H_{n-\min(s_1, s_2)}\bigr).
\end{align*}

Substituting both structural limits into \eqref{eq:exp_sum_identity} completes the proof.
\end{IEEEproof}

\subsection{Combinatorial Cuts from Mutual Exclusivity}
\label{ssec:polytope_bound}

The basic bounds of the previous subsection constrain $E_1$ and $E_2$ individually or through their sum. We now derive a family of bounds that couple $E_1$ and $E_2$ directly, capturing the geometric competition between the two files for column subsets.

The key observation is that for any $s < k$, no $s$-subset of columns can recover both files simultaneously.

\begin{lemma}\label{lem:mutual_excl}
For any rank-$k$ generator matrix $G \in \F_q^{k \times n}$, any partition $s_1 + s_2 = k$, and any $s < k$, the events 
$$\{F_1 \subseteq \langle g_j : j \in S\rangle\}$$ 
and 
$$\{F_2 \subseteq \langle g_j : j \in S\rangle\}$$ are mutually exclusive over all $s$-subsets $S \subseteq [n]$. Consequently,
\begin{equation}\label{eq:mutual_excl}
    \alpha_{F_1}(s) + \alpha_{F_2}(s) \;\leq\; \binom{n}{s} 
    \quad \text{for all } s < k,
\end{equation}
or equivalently, in terms of the normalized probabilities 
$p_i(s) \coloneqq \alpha_{F_i}(s)/\binom{n}{s}$:
\begin{equation}\label{eq:mutual_excl_prob}
    p_1(s) + p_2(s) \;\leq\; 1 \quad \text{for all } s < k.
\end{equation}
\end{lemma}

\begin{IEEEproof}
Suppose for contradiction that some $s$-subset $S \subseteq [n]$ with $s < k$ satisfies both $F_1 \subseteq \langle g_j : j \in S\rangle$ and $F_2 \subseteq \langle g_j : j \in S\rangle$. Then since $F_1 \oplus F_2 = \F_q^k$, we have 
$\F_q^k \subseteq \langle g_j : j \in S\rangle$, which requires 
$\dim(\langle g_j : j \in S\rangle) \geq k$. However, 
$$\dim(\langle g_j : j \in S\rangle) \leq |S| = s < k,$$ a contradiction. 
Thus no such $S$ exists, and the events are mutually exclusive. Equation~\eqref{eq:mutual_excl} follows by counting, and~\eqref{eq:mutual_excl_prob} follows by dividing by $\binom{n}{s}$.
\end{IEEEproof}

To convert this combinatorial constraint into a bound on the retrieval times, we use the following monotonicity property of $p_i(s)$.

\begin{lemma}\label{lem:monotone_p}
For any rank-$k$ generator matrix $G \in \F_q^{k \times n}$, any $I \subseteq [k]$, and any $1 \leq s \leq n-1$:
\begin{equation}\label{eq:monotone_p}
    p_I(s) \;\leq\; p_I(s+1),
\end{equation}
where $p_I(s) = \alpha_I(s)/\binom{n}{s}$ is the probability that a uniformly chosen $s$-subset of $[n]$ recovers all strands in $I$.
\end{lemma}

\begin{IEEEproof}
Count pairs $(S, S')$ where $S \subset S'$, $|S|=s$, $|S'|=s+1$, and $S$ recovers $I$. Since recovery is monotone under inclusion, $S' = S \cup \{j\}$ recovers $I$ for any $j \in [n]\setminus S$, giving exactly $\alpha_I(s)\cdot(n-s)$ such pairs. On the other hand, every recovering $(s+1)$-subset $S'$ contributes at most $s+1$ such pairs (one for each of its $s$-subsets, regardless of whether they recover $I$), giving at most $\alpha_I(s+1)\cdot(s+1)$ such pairs. 
Therefore:
\[
    \alpha_I(s)\cdot(n-s) \;\leq\; \alpha_I(s+1)\cdot(s+1).
\]
Dividing both sides by $\binom{n}{s}\cdot(n-s) = \binom{n}{s+1}\cdot(s+1)$ yields $p_I(s) \leq p_I(s+1)$.
\end{IEEEproof}

We can now state and prove the main result of this subsection.
\begin{theorem}\label{thm:polytope_bound}
For any rank-$k$ generator matrix $G \in \F_q^{k \times n}$, any partition $s_1 + s_2 = k$, and any integer $s^* \in \{\max(s_1,s_2), \ldots, k-1\}$, the achievable expected retrieval 
times satisfy:
\begin{equation}\label{eq:polytope_cut}
    \frac{E_1}{A_1(s^*)} + \frac{E_2}{A_2(s^*)} 
    \;\geq\; \frac{B(s^*)}{A_1(s^*)} + \frac{B(s^*)}{A_2(s^*)} - 1,
\end{equation}
where
\begin{align}
    A_i(s^*) &\;\coloneqq\; \sum_{s=s_i}^{s^*} \frac{n}{n-s}, 
    \label{eq:Ai_def}\\
    B(s^*)   &\;\coloneqq\; n\harm{n} - \sum_{s=s^*+1}^{n-1} 
    \frac{n}{n-s}.\label{eq:B_def}
\end{align}
\end{theorem}

\begin{IEEEproof}
First, we rewrite $E_i$ in terms of $p_i(s)$.
By Lemma~\ref{lem:exp_formula_subset} and Definition~\ref{def:alpha_subset}, recalling that $p_i(s) = \alpha_{F_i}(s)/\binom{n}{s}$ and that $\alpha_{F_i}(s) = 0$ for $s < s_i$:
\begin{equation}\label{eq:Ei_rewrite}
    E_i \;=\; n\harm{n} - \sum_{s=s_i}^{n-1} 
    \frac{\alpha_{F_i}(s)}{\binom{n-1}{s}}
    \;=\; n\harm{n} - \sum_{s=s_i}^{n-1} p_i(s) \cdot \frac{n}{n-s}.
\end{equation}
where in the last equality we used $\frac{\binom{n}{s}}{\binom{n-1}{s}} = \frac{n}{n-s}$.

Next, we lower bound $p_i(s^*)$ in terms of $E_i$ as follows. Fix $s^* \in \{\max(s_1,s_2), \ldots, k-1\}$.  By Lemma~\ref{lem:monotone_p}, $p_i(s) \leq p_i(s^*)$ for all $s \leq s^*$, and trivially $p_i(s) \leq 1$ for all $s$. Applying these bounds 
to~\eqref{eq:Ei_rewrite}:
\begin{align}
E_i &\;=\; n\harm{n} - \sum_{s=s_i}^{s^*} p_i(s)\cdot\frac{n}{n-s} 
- \sum_{s=s^*+1}^{n-1} p_i(s)\cdot\frac{n}{n-s} \notag\\
&\;\geq\; n\harm{n} - p_i(s^*)\sum_{s=s_i}^{s^*}\frac{n}{n-s} 
- \sum_{s=s^*+1}^{n-1}\frac{n}{n-s} \notag\\
&\;=\; B(s^*) - p_i(s^*) \cdot A_i(s^*), \label{eq:Ei_lower}
\end{align}
where in the second line we applied $p_i(s) \leq p_i(s^*)$ to the first sum and $p_i(s) \leq 1$ to the second sum, and in the last line we substituted with $A(s^*),B(s^*)$ by their definition. Rearranging 
\eqref{eq:Ei_lower} implies that 
\begin{equation}\label{eq:pi_lower}
    p_i(s^*) \;\geq\; \frac{B(s^*) - E_i}{A_i(s^*)}.
\end{equation}

Lastly, we apply the mutual exclusivity. Since $s^* < k$, Lemma~\ref{lem:mutual_excl} gives $p_1(s^*) + p_2(s^*) \leq 1$. Substituting~\eqref{eq:pi_lower} for both files yields
\begin{equation}
    \frac{B(s^*) - E_1}{A_1(s^*)} + \frac{B(s^*) - E_2}{A_2(s^*)} 
    \;\leq\; 1.
\end{equation}
Rearranging gives~\eqref{eq:polytope_cut}.
\end{IEEEproof}

\begin{remark}
The range $s^* \geq \max(s_1, s_2)$ is required to ensure $A_i(s^*) > 0$ for both $i \in \{1,2\}$, so that~\eqref{eq:polytope_cut} is well-defined. 
Specifically, $A_i(s^*) = 0$ whenever $s^* < s_i$. The upper bound $s^* \leq k-1$ is required for Lemma~\ref{lem:mutual_excl} to apply. 
Note that for asymmetric partitions with $s_1 \neq s_2$, the valid range $\{\max(s_1,s_2), \ldots, k-1\}$ may contain only a single value; in particular, when $\max(s_1,s_2) = k-1$ (e.g., $s_1=1$, $s_2=k-1$), the only valid choice is $s^* = k-1$.
\end{remark}

Taking $s^* = k-1$ gives the tightest cut within the family, as stated in the following corollary. Using $s^* = k-1$ in~\eqref{eq:B_def}:
\begin{align}
    B(k-1) &\;=\; n\harm{n} - \sum_{s=k}^{n-1} \frac{n}{n-s} \notag\\
           &\;=\; n\harm{n} - n\sum_{j=1}^{n-k} \frac{1}{j} \notag\\
           &\;=\; n\bigl(\harm{n} - \harm{n-k}\bigr),
\end{align}
where in the second equality we substituted $j = n-s$. Note that $n(\harm{n} - \harm{n-k})$ is precisely the expected number of draws required to recover all $k$ information strands under an $[n,k]$ MDS code~\cite{bar-lev2025cover}, reflecting the coupon-collector structure of full-code recovery.

\begin{corollary}\label{cor:tightest_cut}
Under the same conditions as Theorem~\ref{thm:polytope_bound}, with $B^* \coloneqq n(\harm{n} - \harm{n-k})$ and $A_i^* \coloneqq \sum_{s=s_i}^{k-1}\frac{n}{n-s}$, the cut at $s^* = k-1$ dominates all other cuts in the family and reads:
\begin{equation}\label{eq:tightest_cut}
    \frac{E_1}{A_1^*} + \frac{E_2}{A_2^*} 
    \;\geq\; B^*\!\left(\frac{1}{A_1^*} + \frac{1}{A_2^*}\right) - 1.
\end{equation}
\end{corollary}

\begin{IEEEproof}
Define $\beta_i \coloneqq n(\harm{n} - \harm{n-s_i})$ for $i \in \{1,2\}$.
We first observe that for any $${s^* \in \{\max(s_1,s_2),\ldots,k-1\}}$$ we have that
\begin{align}
    B(s^*) - A_i(s^*) 
    \;&=\; n\harm{n} - \sum_{s=s^*+1}^{n-1}\frac{n}{n-s} \nonumber
    - \sum_{s=s_i}^{s^*}\frac{n}{n-s}
    \;\\&=\; n\harm{n} - \sum_{s=s_i}^{n-1}\frac{n}{n-s}
    \;=\; \beta_i,\label{eq:beta_const}
\end{align}
which is independent of $s^*$. Next we show that the cut at $s^*$ from Theorem~\ref{thm:polytope_bound} is equivalent to
\begin{equation}\label{eq:cut_rewritten}
    \frac{E_1 - \beta_1}{A_1(s^*)} + \frac{E_2 - \beta_2}{A_2(s^*)} 
    \;\geq\; 1.
\end{equation}

By~\eqref{eq:beta_const}, $B(s^*) =\beta_i + A_i(s^*)$ for each $i \in \{1,2\}$. Subtracting $\frac{B(s^*)}{A_1(s^*)} + \frac{B(s^*)}{A_2(s^*)}$ from both sides of~\eqref{eq:polytope_cut} and substituting $B(s^*) = \beta_i + A_i(s^*)$ into each numerator gives:
\begin{align*}
    \frac{E_1-B(s^*)}{A_1(s^*)} + \frac{E_2-B(s^*)}{A_2(s^*)} 
    &\;\geq\; -1 \\
    \frac{E_1-\beta_1-A_1(s^*)}{A_1(s^*)} 
    + \frac{E_2-\beta_2-A_2(s^*)}{A_2(s^*)} 
    &\;\geq\; -1 \\
    \frac{E_1-\beta_1}{A_1(s^*)} + \frac{E_2-\beta_2}{A_2(s^*)} 
    - 2 &\;\geq\; -1,
\end{align*}
which yields~\eqref{eq:cut_rewritten}.

We next show that $E_i \geq \beta_i$ for any rank-$k$ generator matrix~$G$. Since $\alpha_{F_i}(s) = 0$ for $s < s_i$ and $\alpha_{F_i}(s) \leq \binom{n}{s}$ for all $s$, we have:
\begin{align}
    E_i \;&=\; n\harm{n} - \sum_{s=s_i}^{n-1} 
    \frac{\alpha_{F_i}(s)}{\binom{n-1}{s}}
    \;\geq\; n\harm{n} - \sum_{s=s_i}^{n-1} \frac{n}{n-s}
    \;\nonumber
    \\&=\; n(\harm{n} - \harm{n-s_i}) \;=\; \beta_i,\label{eq:Ei_beta_lb}
\end{align}
where we used $\binom{n}{s}/\binom{n-1}{s} = n/(n-s)$ and the substitution $j = n-s$ to obtain $\sum_{s=s_i}^{n-1}\frac{n}{n-s} = n\harm{n-s_i}$.

Now fix any $s' \in \{\max(s_1,s_2),\ldots,k-2\}$. Since $s' < k-1$, we have $A_i(s') < A_i(k-1) = A_i^*$ (both $A_i(x)$ are strictly increasing in $x$). Since $E_i - \beta_i \geq 0$ by \eqref{eq:Ei_beta_lb}, each term in~\eqref{eq:cut_rewritten} satisfies:
\[
    \frac{E_i - \beta_i}{A_i(s')} 
    \;\geq\; \frac{E_i - \beta_i}{A_i^*}.
\]
Summing over $i \in \{1,2\}$ gives:
\[
    \frac{E_1-\beta_1}{A_1(s')} + \frac{E_2-\beta_2}{A_2(s')}
    \;\geq\; \frac{E_1-\beta_1}{A_1^*} + \frac{E_2-\beta_2}{A_2^*}.
\]
Therefore, if $(E_1,E_2)$ satisfies the $s^*=k-1$ cut, i.e., the right-hand side above is at least~$1$, then the left-hand side is also  at least $1$, which is exactly the cut at $s'$. This completes the proof.
\end{IEEEproof}

\begin{remark}
The bound $E_i \geq \beta_i = n(\harm{n} - \harm{n-s_i})$, established in~\eqref{eq:Ei_beta_lb}, is strictly stronger than Proposition~\ref{prop:basic_lb} for all finite $n$ and $s_i \geq 1$, since $n(\harm{n} - \harm{n-s_i}) \geq s_i$ with equality only when $n = s_i$ (which is not possible for two files). Thus Proposition~\ref{prop:basic_lb} is subsumed by Corollary~\ref{cor:tightest_cut}.
\end{remark}

\Cref{fig:cuts_n20,fig:cuts_n50} illustrate the family of cuts  for $k=8$ and $n=20, 50$ respectively, across four partition configurations. The tightest cut ($s^*=k-1$, solid line) uniformly dominates all weaker cuts (dashed lines) across the entire $(E_1,E_2)$ plane. The black dots indicate the operating points  of the local MDS strategy. Observe that as $n$ grows, the cuts tighten and the local MDS points move closer to the cut boundaries, consistent with the asymptotic analysis of 
\cref{sec:asymptotics}.

\begin{figure*}[h]
\centering
\includegraphics[width=\linewidth]{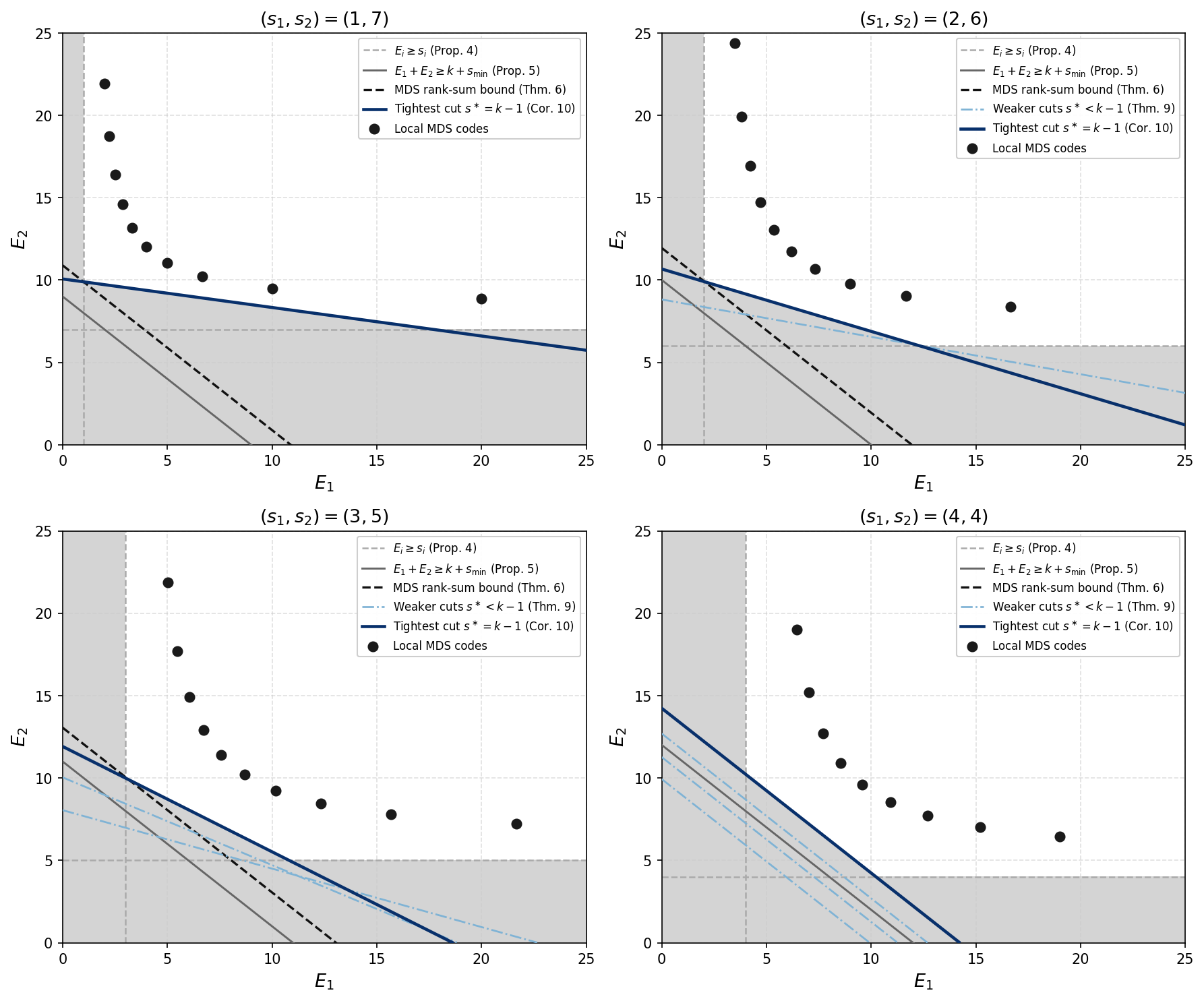}
\caption{Combinatorial cut bounds for $k=8$, $n=20$, across 
four partition configurations $(s_1,s_2)$. The solid line is 
the tightest cut $s^*=k-1$; dashed lines show weaker cuts. 
Black dots are local MDS operating points.}
\label{fig:cuts_n20}
\end{figure*}

\begin{figure*}[h]
\centering
\includegraphics[width=\linewidth]{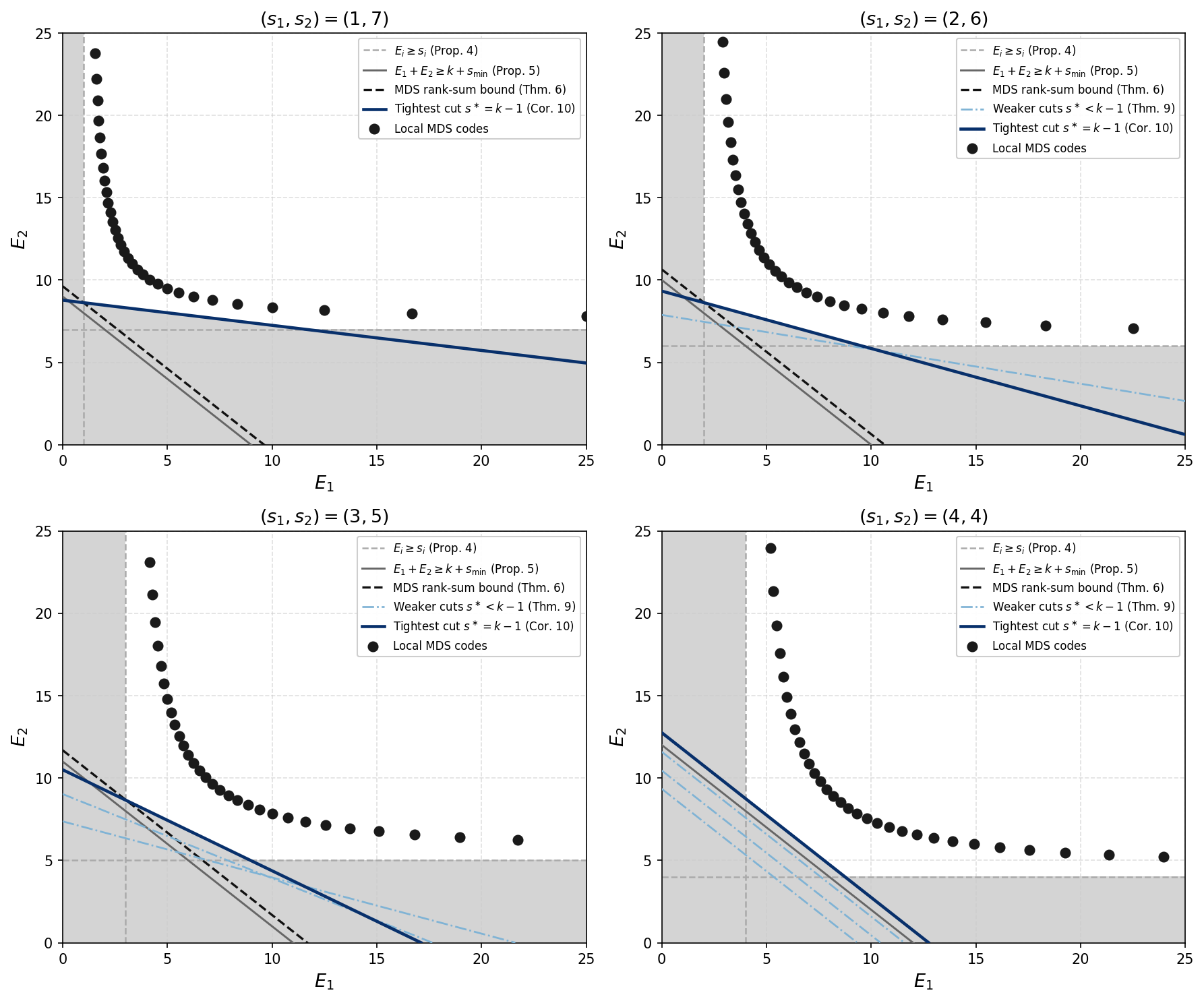}
\caption{Same as \cref{fig:cuts_n20} but for $n=50$. 
The cuts tighten noticeably and the local MDS points  approach the cut boundaries.}
\label{fig:cuts_n50}
\end{figure*}

\begin{remark}\label{rem:linear_bounds_polytope}
All bounds established in this section (Propositions~\ref{prop:basic_lb} and~\ref{prop:joint_lb}, Theorem~\ref{thm:mds_rank_sum}, and the family of cuts in Theorem~\ref{thm:polytope_bound} and Corollary~\ref{cor:tightest_cut}) are \emph{linear} in the pair $(E_1, E_2)$. Collectively they define a convex polytope that outer-approximates the true non-achievable region. As $n$~grows (Figures~\ref{fig:cuts_n20} and~\ref{fig:cuts_n50}), the local MDS operating points accumulate along a curve of the form $s_1/E_1 + s_2/E_2 = \mathrm{const}$, suggesting that the fundamental trade-off between $E_1$ and $E_2$ is governed by a {hyperbolic} constraint. No finite collection of linear cuts can describe such a boundary exactly. Proving a nonlinear lower bound of this form requires a qualitatively different approach, which we develop in the following section.
\end{remark}

\section{A Nonlinear Bound via Column Geometry}
\label{sec:nonlinear_bound}

While the finite collection of linear cuts derived in Section~\ref{sec:lower_bounds} bounds the local MDS operating points, no purely linear system can natively describe the hyperbolic trade-off curve they form. 
Before exploiting the internal column structure of the generator matrix, we record what can be said about the nonlinear quantity $s_1/E_1 + s_2/E_2$ from the polytope constraints
of Section~\ref{sec:lower_bounds} alone. The next two corollaries give unconditional
upper bounds on this quantity --- valid for \emph{every} linear code, with no
assumption on its column structure --- in explicit closed form. They are the
strongest bounds obtainable from the aggregate constraints of Propositions~\ref{prop:basic_lb} and~\ref{prop:joint_lb}, and thus the best provable, assumption-free counterparts to the hyperbolic bound. The column-geometry argument
that follows sharpens them to $\le 1$ under a structural condition and motivates the
universal conjecture.

\begin{corollary}\label{cor:smax_bound}
For any $(E_1, E_2) \in \Lam_{n,k}(s_1, s_2)$,
\begin{equation}\label{eq:smax_bound}
    \frac{s_1}{E_1} + \frac{s_2}{E_2} \;\le\; 1 + \frac{s_{\max}}{k},
\end{equation}
where $s_{\max} \coloneqq \max\{s_1, s_2\}$.
\end{corollary}
\begin{IEEEproof}
Without loss of generality assume ${s_1 \le s_2}$, so $s_{\min} = s_1$ and $s_{\max} = s_2$. Let 
$$f(E_1, E_2) \coloneqq \frac{s_1}{E_1} + \frac{s_2}{E_2}.$$ 
By Propositions~\ref{prop:basic_lb} and~\ref{prop:joint_lb}, every $(E_1, E_2) \in \Lam_{n,k}(s_1,s_2)$ lies in the polytope
\[
    \mathcal{P} \;\coloneqq\; \left\{(E_1,E_2)\in\mathbb{R}_{>0}^2 
    : \substack{E_1 \ge s_1,\; E_2 \ge s_2,\;\\ E_1+E_2 \ge k+s_1}\right\}.
\]
The Pareto-minimal boundary of $\mathcal{P}$, i.e., the set of points $(E_1, E_2) \in \mathcal{P}$ such that no $(E_1', E_2') \in \mathcal{P}$ satisfies $E_i' \le E_i$ for both $i$ with at least one strict inequality, consists of exactly three pieces. The diagonal constraint $E_1 + E_2 = k + s_1$ intersects the vertical wall $E_1 = s_1$ at $(s_1, k)$ and the horizontal floor $E_2 = s_2$ at $(2s_1, s_2)$ (as substituting $E_2 = s_2 = k - s_1$ into $E_1 + E_2 = k + s_1$ gives $E_1 = k + s_1 - (k - s_1) = 2s_1$). This yields the finite segment $$\{E_1+E_2 = k+s_1,\; E_1\in[s_1,2s_1]\},$$ the ray 
$$\{E_1 = s_1,\; E_2 \ge k\}$$ 
extending upward from $(s_1,k)$, and the ray $$\{E_2 = s_2,\; E_1 \ge 2s_1\}$$ extending rightward from $(2s_1,s_2)$.

Since $f$ is strictly decreasing in each variable separately, its supremum over $\mathcal{P}$ is attained on this Pareto-minimal boundary. On the left ray, $f = 1 + s_2/E_2$ is decreasing in $E_2$, so its maximum is at $E_2 = k$, the left endpoint of the segment. On the bottom ray, $f = s_1/E_1 + 1$ is decreasing in $E_1$, so its maximum is at $E_1 = 2s_1$, the right endpoint of the segment. Hence the global maximum of $f$ over $\mathcal{P}$ is attained on the finite segment.

On the segment, substitute $E_2 = k + s_1 - E_1$ to obtain the single-variable function
\[
    h(E_1) \;\coloneqq\; \frac{s_1}{E_1} + \frac{s_2}{k+s_1-E_1}, 
    \qquad E_1 \in [s_1, 2s_1].
\]
Its second derivative is
\[
    h''(E_1) \;=\; \frac{2s_1}{E_1^3} + \frac{2s_2}{(k+s_1-E_1)^3},
\]
which is strictly positive since $E_1 > 0$ and $k+s_1-E_1 = E_2 \ge s_2 > 0$ throughout the interval. Hence $h$ is strictly convex on $[s_1, 2s_1]$, and its maximum is attained at one of the two endpoints. Evaluating:
\begin{align*}
    h(s_1) &\;=\; \frac{s_1}{s_1} + \frac{s_2}{k} 
    \;=\; 1 + \frac{s_2}{k}, \\[4pt]
    h(2s_1) &\;=\; \frac{s_1}{2s_1} + \frac{s_2}{s_2} 
    \;=\; \frac{1}{2} + 1 \;=\; \frac{3}{2}.
\end{align*}
Since $s_1 \le s_2$ and $s_1 + s_2 = k$, we have $s_2 \ge k/2$, hence $s_2/k \ge 1/2$, so $1 + s_2/k \ge 3/2$, with equality if and only if $s_1 = s_2$. The maximum of $h$ is therefore 
$1 + s_2/k = 1 + s_{\max}/k$, which completes the proof.
\end{IEEEproof}

\begin{corollary}\label{cor:cs_bound}
For any $(E_1, E_2) \in \Lam_{n,k}(s_1, s_2)$,
\begin{equation}\label{eq:cs_bound}
    \frac{s_1}{E_1} + \frac{s_2}{E_2} 
    \;\le\; \frac{k^2}{2s_{\min}^2 + s_{\max}^2}.
\end{equation}
For the symmetric partition $s_1 = s_2 = k/2$, this evaluates to~$4/3$. For the extreme asymmetric partition $s_{\min} = 1$, $s_{\max} = k-1$, it evaluates to $k^2/(k^2 - 2k + 3)$.
\end{corollary}

\begin{IEEEproof}
Without loss of generality, assume $s_1 \le s_2$. By the Cauchy-Schwarz inequality applied to the vectors 
$$\left(\sqrt{\frac{s_1}{E_1}}, \sqrt{\frac{s_2}{E_2}}\right) \text{ and } \left(\sqrt{s_1 E_1}, \sqrt{s_2 E_2}\right)$$ we have that 
\[
    \left(\frac{s_1}{E_1} + \frac{s_2}{E_2}\right)
    \!\left(s_1 E_1 + s_2 E_2\right)
    \;\ge\; (s_1 + s_2)^2 \;=\; k^2,
\]
which gives 
$\frac{s_1}{E_1} + \frac{s_2}{E_2} \le \frac{k^2}{s_1 E_1 + s_2 E_2}$. 
Since $k^2/x$ is decreasing in $x$, it remains to find the minimum of $$g(E_1, E_2) \coloneqq s_1 E_1 + s_2 E_2$$ over $\mathcal{P}$.
Since $g$ is linear, its minimum over each piece of the Pareto-minimal boundary is at an endpoint of that piece. On the left ray $\{E_1 = s_1,\, E_2 \ge k\}$, $g = s_1^2 + s_2 E_2$ is increasing in $E_2$, so the minimum is at $E_2 = k$, giving the point $(s_1, k)$. On the bottom ray $\{E_2 = s_2,\, E_1 \ge 2s_1\}$, $g = s_1 E_1 + s_2^2$ is increasing in $E_1$, so the minimum is at  $E_1 = 2s_1$, giving the point $(2s_1, s_2)$. On the diagonal 
segment, substituting ${E_2 = k + s_1 - E_1}$ gives $g = (s_1 - s_2)E_1 + s_2(k+s_1)$, which is linear with slope 
$s_1 - s_2 \le 0$, hence decreasing in $E_1$, so again the minimum 
is at $E_1 = 2s_1$ (the second point). Thus the global minimum of $g$ over 
$\mathcal{P}$ is attained at one of these two points. Evaluating:
\begin{align*}
    g(s_1,\, k) 
    &\;=\; s_1^2 + s_2 k 
    \;=\; s_1^2 + s_1 s_2 + s_2^2, \\
    g(2s_1,\, s_2) 
    &\;=\; 2s_1^2 + s_2^2.
\end{align*}
The difference is 
$g(s_1, k) - g(2s_1, s_2) = s_1(s_2 - s_1) \ge 0$, so the minimum is $2s_1^2 + s_2^2 = 2s_{\min}^2 + s_{\max}^2$, 
attained at $(2s_1,s_2)$. Substituting into the Cauchy-Schwarz 
bound yields~\eqref{eq:cs_bound}.

For the symmetric case $s_1 = s_2 = k/2$ we have that 
$$2\left(\frac{k}{2}\right)^2 + \left(\frac{k}{2}\right)^2 = \frac{3k^2}{4},$$ giving the bound $4/3$. 

For the extreme asymmetric case $s_{\min}=1$, $s_{\max}=k-1$ we have that
$$2 + (k-1)^2 = k^2 - 2k + 3,$$ giving the bound 
$k^2/(k^2-2k+3)$.
\end{IEEEproof}

\begin{remark}
Corollary~\ref{cor:smax_bound} provides a simple and intuitive bound. That is, the trade-off between $E_1$ and $E_2$ is governed by the relative size of the larger file. However, Corollary~\ref{cor:cs_bound} is uniformly tighter, i.e., it dominates Corollary~\ref{cor:smax_bound} for all valid partitions $(s_1, s_2)$, but requires the Cauchy-Schwarz argument and is less immediately interpretable. 
For symmetric partitions ($s_1 = s_2$), the gap between the two bounds is moderate: $4/3$ versus $3/2$. For highly asymmetric partitions ($s_{\min} = 1$, $s_{\max} = k-1$), the gap is 
dramatic: Corollary~\ref{cor:cs_bound} gives $k^2/(k^2-2k+3) \to 1$ as $k \to \infty$, while Corollary~\ref{cor:smax_bound} gives only $2 - 1/k \to 2$. 
\end{remark}

While these bounds already go beyond the linear polytope description, they depend only on the aggregate constraints of Propositions~\ref{prop:basic_lb} and~\ref{prop:joint_lb} and do not exploit the internal column structure of the code. In the rest of this section we develop a geometric approach based on projecting the drawn columns onto each file subspace. The key structural insight is that a column $g_j$ can contribute to the recovery of $F_i$ only if its $F_i$-component is nonzero, and the number of such columns is directly controlled by the pure-file column counts $N(F_1)$ and $N(F_2)$. The bound simplifies substantially for codes whose columns all belong entirely to one of the two files, and the resulting expression motivates a universal conjecture for the general case, which we state and discuss at the end of the section.

\subsection{Projected retrieval times and the main theorem}

Since $F_1 \oplus F_2 = \F_q^k$, every column $g_j$ of $G$ decomposes uniquely as $g_j = \pi_1(g_j) + \pi_2(g_j)$, where $\pi_i \colon \F_q^k \to F_i$ is the projection onto $F_i$ along $F_{3-i}$. Concretely, $\pi_1$ retains only the first $s_1$ coordinates and $\pi_2$ retains only the last $s_2$ coordinates. Observe that $\pi_i(g_j) = 0$ if and only if $g_j \in F_{3-i}$, that is, $g_j$ is a pure-$F_{3-i}$ column. In particular, the number of columns with \emph{nonzero} $F_i$-projection is exactly $n - N(F_{3-i})$, where $N(\cdot)$ is the column counting function of Definition~\ref{def:column_count}.

Let $\xi_1, \xi_2, \ldots \in [n]$ denote the sequence of independently and uniformly drawn column indices.

\begin{definition}[Projected retrieval time]\label{def:proj_tau} For $i \in \{1,2\}$, the {projected retrieval time} $\tau_i^{\mathrm{proj}}$ is the first time $t$ at which the projected columns span $F_i$:
\[
    \tau_i^{\mathrm{proj}} 
    \;\coloneqq\; 
    \min\bigl\{t \ge 0 : 
        F_i \subseteq \langle \pi_i(g_{\xi_r}) \text{ such that } r \le t \rangle
    \bigr\}.
\]
\end{definition}

\begin{lemma}\label{lem:tau_proj_lb}
For any rank-$k$ generator matrix $G$ and ${i \in \{1,2\}}$, we have that $\taufile{F_i}(G) \;\ge\; \tau_i^{\mathrm{proj}}$ for every realization of the draw sequence, and hence $E_i(G) \;\ge\; \E[\tau_i^{\mathrm{proj}}]$.
\end{lemma}
\begin{IEEEproof}
We show that for every realization of the draw sequence, every time $t$ at which $F_i$ is recovered is also a time at which $\tau_i^{\mathrm{proj}}$ has already been reached.

Fix any realization of the draw sequence $\xi_1, \xi_2, \ldots$ and any time $t \ge 0$. Suppose that $t \ge \taufile{F_i}(G)$, meaning $F_i \subseteq V_t$ where $V_t \coloneqq \langle g_{\xi_r} : r \le t \rangle$.

Since $\pi_i \colon \F_q^k \to F_i$ is a linear map, it maps spans to spans, i.e., for any set of vectors $\{v_1,\ldots,v_m\}$,
\[
    \pi_i\bigl(\langle v_1,\ldots,v_m \rangle\bigr) 
    \;=\; \langle \pi_i(v_1),\ldots,\pi_i(v_m) \rangle.
\]
Applying this to $V_t$ gives $\pi_i(V_t) = \langle \pi_i(g_{\xi_r}) \text{ such that } r \le t \rangle$. Since $F_i \subseteq V_t$ and $\pi_i$ is linear and as such, preserves inclusion, applying $\pi_i$ to both sides of this inclusion gives $\pi_i(F_i) \subseteq \pi_i(V_t)$. Now, recall that $\pi_i$ is the projection onto $F_i$ along $F_{3-i}$, meaning every vector $v \in \F_q^k$ decomposes as $v = \pi_i(v) + \pi_{3-i}(v)$ with $\pi_i(v) \in F_i$. In particular, for any $v \in F_i$ the $F_{3-i}$-component is zero, so $\pi_i(v) = v$. Therefore, $\pi_i$ fixes $F_i$ pointwise, giving $\pi_i(F_i) = F_i$. Hence, we have that 
\[
    F_i 
    \;=\; \pi_i(F_i) 
    \;\subseteq\; \pi_i(V_t) 
    \;=\; \langle \pi_i(g_{\xi_r}) \text{ such that } r \le t \rangle,
\]
which is precisely the condition $\tau_i^{\mathrm{proj}} \le t$ by Definition~\ref{def:proj_tau}. Since the realization and $t$ were arbitrary, we have established that for every realization:
\[
    t \ge \taufile{F_i}(G) \;\implies\; t \ge \tau_i^{\mathrm{proj}},
\]
which is equivalent to $\taufile{F_i}(G) \ge \tau_i^{\mathrm{proj}}$ for every realization. Taking expectations preserves the inequality, giving $E_i(G) \ge \E[\tau_i^{\mathrm{proj}}]$.
\end{IEEEproof}

We analyze $\tau_i^{\mathrm{proj}}$ by decomposing it into $s_i$ sequential phases. For $0 \le \ell \le s_i - 1$, define phase $\ell$ as the period during which the projected span has dimension exactly $\ell$, and let $W_\ell$ be the number of draws in phase $\ell$. Since the projected columns must span $F_i$, exactly $s_i$ dimension increments must occur, and $\tau_i^{\mathrm{proj}} = \sum_{\ell=0}^{s_i-1} W_\ell$.

Within phase $\ell$, a draw of column $g_j$ is \emph{useful} --- meaning it increments the dimension --- if and only if $\pi_i(g_j)$ lies outside the current projected span. A necessary condition for $g_j$ to be useful in any phase is that $\pi_i(g_j) \ne 0$, since the current span always contains the zero vector. Now $\pi_i(g_j) = 0$ if and only if $g_j \in F_{3-i}$, and there are exactly $N(F_{3-i})$ such columns. Therefore, regardless of the current projected span, the number of potentially useful columns is at most $n - N(F_{3-i})$.

\begin{lemma}\label{lem:proj_lower}
For any rank-$k$ generator matrix $G$ and ${i \in \{1,2\}}$ we have that
\begin{equation}\label{eq:proj_lower}
    \E[\tau_{i}^{\mathrm{proj}}] \;\ge\; \frac{n\,s_i}{n - N(F_{3-i})}.
\end{equation}
\end{lemma}

\begin{IEEEproof}
We first verify that $\tau_i^{\mathrm{proj}}$ is finite almost surely. Since $G$ has rank $k$ and $\pi_i \colon \F_q^k \to F_i$ is surjective, the projected columns $\{\pi_i(g_j)\}_{j=1}^n$ span $F_i$. Therefore, for any $\ell \in \{0,\ldots,s_i-1\}$, there exists at least one column $g_j$ such that $\pi_i(g_j)$ lies outside the current projected span, meaning the probability of making progress in each phase is bounded away from zero by $1/n$. Since each phase terminates in finite expected time, all $s_i$ phases terminate almost surely, and hence $\tau_i^{\mathrm{proj}} < \infty$ almost surely.

We now lower bound $\E[\tau_i^{\mathrm{proj}}]$. As discussed above, we have that
$\tau_i^{\mathrm{proj}} = \sum_{\ell=0}^{s_i-1} W_\ell$, where $W_\ell$ 
is the number of draws in phase $\ell$. Fix a phase $\ell$ and let $P_\ell^{(i)}$ denote the current projected span at the start of phase $\ell$, which has dimension $\ell$. A draw of column $g_j$ is useful in phase $\ell$ if and only if $\pi_i(g_j) \notin P_\ell^{(i)}$. As argued above, a necessary condition for usefulness is $\pi_i(g_j) \ne 0$, since $0 \in P_\ell^{(i)}$. Therefore the number of useful columns satisfies
\begin{align*}
    \hat{n}_i(\ell) 
    \;& \coloneqq\; \bigl|\bigl\{j \in [n] : \pi_i(g_j) \notin P_\ell^{(i)}\bigr\}\bigr|
    \; \le \;
    n - N(F_{3-i}).
\end{align*}

Since each draw is uniform over $[n]$ independently, conditional on $P_\ell^{(i)}$ the phase $\ell$ waiting time $W_\ell$ is a geometric random variable with success probability $\hat{n}_i(\ell)/n$, giving
\[
    \E\bigl[W_\ell \;\big|\; P_\ell^{(i)}\bigr] 
    \;=\; \frac{n}{\hat{n}_i(\ell)} 
    \;\ge\; \frac{n}{n - N(F_{3-i})}.
\] 
Recall that the tower property of conditional expectation states 
that $\E[\E[X \mid Y]] = \E[X]$ for any integrable random variable $X$ 
and any random variable $Y$. As  $W_\ell$ is geometric with success probability $\hat{n}_i(\ell)/n \ge 1/n > 0$ it is  integrable and hence, by taking expectations of both sides we have
\begin{align*}
    \E[W_\ell] \;&=\; \E\Bigl[\E\bigl[W_\ell \;\big|\; P_\ell^{(i)}\bigr]\Bigr] 
    \;\ge\; \E\left[\frac{n}{n - N(F_{3-i})}\right]
    \;\\ & =\; \frac{n}{n - N(F_{3-i})},
\end{align*}
where the last equality holds because 
$n/(n-N(F_{3-i}))$ is a constant.

Summing over all $s_i$ phases and using linearity of expectation:
\begin{align*}
    \E[\tau_i^{\mathrm{proj}}] 
    \;& =\; \sum_{\ell=0}^{s_i-1} \E[W_\ell] 
    \;\ge\; \sum_{\ell=0}^{s_i-1} \frac{n}{n-N(F_{3-i})} 
    \;\\ & =\; \frac{n\,s_i}{n - N(F_{3-i})},
\end{align*}
which concludes the proof. 
\end{IEEEproof}

\begin{theorem}\label{thm:nonlinear_bound}
For any rank-$k$ generator matrix $G \in \F_q^{k \times n}$ and any 
partition $s_1 + s_2 = k$ we have that 
\begin{equation}\label{eq:nonlinear_bound}
    \frac{s_1}{E_1(G)} + \frac{s_2}{E_2(G)} 
    \;\le\; 2 - \frac{N(F_1) + N(F_2)}{n}.
\end{equation}
\end{theorem}
\begin{IEEEproof}
By Lemmas~\ref{lem:tau_proj_lb} and~\ref{lem:proj_lower}, for each 
$i \in \{1,2\}$ we have 
\[
E_i(G) \;\ge\; \frac{n\,s_i}{n - N(F_{3-i})},\]
and hence
\[
    \frac{s_i}{E_i(G)} \;\le\; \frac{n - N(F_{3-i})}{n}.
\]
Summing over $i \in \{1,2\}$ gives
\begin{align*}
    \frac{s_1}{E_1(G)} + \frac{s_2}{E_2(G)} 
    \;& \le\; \frac{(n - N(F_2)) + (n - N(F_1))}{n}
    \;\\ & =\; 2 - \frac{N(F_1) + N(F_2)}{n}. \qedhere
\end{align*}
\end{IEEEproof}

\begin{corollary}\label{cor:pure_file_hyperbola}
If every column of $G$ belongs entirely to $F_1$ or to $F_2$, then $N(F_1) + N(F_2) = n$, and consequently
\begin{equation}\label{eq:pure_hyperbola}
    \frac{s_1}{E_1(G)} + \frac{s_2}{E_2(G)} \;\le\; 1.
\end{equation}
In particular, the latter holds for any local MDS code.
\end{corollary}
\begin{IEEEproof}
Since $F_1 \cap F_2 = \{0\}$, no nonzero column can lie in both $F_1$ 
and $F_2$, so $N(F_1)$ and $N(F_2)$ count disjoint sets of columns. If 
every column is pure-file, each $j \in [n]$ is counted by exactly one 
of $N(F_1)$, $N(F_2)$, giving $N(F_1) + N(F_2) = n$. 
Substituting into \eqref{eq:nonlinear_bound} yields \eqref{eq:pure_hyperbola}.
\end{IEEEproof}

\begin{remark}
Let $m_{\mathrm{mix}} \coloneqq n - N(F_1) - N(F_2) \ge 0$ denote the 
number of \emph{mixed} columns, i.e., those lying outside $F_1 \cup F_2$. 
Then \eqref{eq:nonlinear_bound} reads 
$s_1/E_1 + s_2/E_2 \le 1 + m_{\mathrm{mix}}/n$: the bound degrades 
proportionally to the fraction of mixed columns, and recovers $\le 1$ 
precisely when $m_{\mathrm{mix}} = 0$.
\end{remark}

\subsection{The universal hyperbolic conjecture}

Theorem~\ref{thm:nonlinear_bound} and Corollary~\ref{cor:pure_file_hyperbola} together suggest that $s_1/E_1 + s_2/E_2 \le 1$ is a universal constraint, not restricted to the pure-file subfamily. We state this as the central conjecture of the paper.

\begin{conjecture}[Universal hyperbolic bound]\label{conj:hyperbola}
For any rank-$k$ generator matrix $G \in \F_q^{k \times n}$ and any partition $s_1 + s_2 = k$ with $\max\{s_1,s_2\} \ge 2$,
\begin{equation}\label{eq:hyperbola_conj}
    \frac{s_1}{E_1(G)} + \frac{s_2}{E_2(G)} \;\le\; 1.
\end{equation}
\end{conjecture}

We provide three pieces of supporting evidence.

\medskip
\noindent\textbf{(i) Verification on explicit constructions and numerical experiments.}
File-dedicated codes are covered analytically by Corollary~\ref{cor:pure_file_hyperbola}. 
For the hybrid code $\cC^*$ of Example~\ref{ex:hybrid_n8}, $E_1 = 403/105$ and $E_2 = 584/105$ with $s_1 = 1$, $s_2 = 3$, giving $s_1/E_1 + s_2/E_2 = 105/403 + 315/584 \approx 0.80 \le 1$. The global systematic MDS code of Example~\ref{ex:asymmetric_n8_local} satisfies the bound with $1/4 + 3/(106/21) \approx 0.84 \le 1$. The general code 
constructions of Section~\ref{sec:constructions} likewise all satisfy Conjecture~\ref{conj:hyperbola}, as we verify explicitly for each family in Corollaries~\ref{cor:identity_conjecture}, \ref{cor:dedicated_conjecture}, and \ref{cor:global_MDS_conjecture}. Beyond these closed-form examples, we have verified Conjecture~\ref{conj:hyperbola} numerically for a large random sample of generator matrices across multiple parameter settings $(n, k, s_1, s_2)$, with no violation observed in any case.

\medskip
\noindent\textbf{(ii) Asymptotic tightness.}
As $n \to \infty$ with fixed allocation ratios $\rho_i = n_i/n$, the local MDS family achieves $E_i \to s_i/\rho_i$, so that $s_1/E_1 + s_2/E_2 \to \rho_1 + \rho_2 = 1$. We establish this formally in Section~\ref{sec:asymptotics}, proving that the local MDS family traces the full boundary curve $\{s_1/E_1 + s_2/E_2 = 1,\, E_i \ge s_i\}$ in the limit. This shows that Conjecture~\ref{conj:hyperbola} is asymptotically tight, in the sense that the bound $\le 1$ cannot be replaced by any strictly smaller constant. Whether the hyperbolic curve is also a lower bound on the asymptotic achievability region, i.e., whether no sequence of codes can approach a point strictly below it, is equivalent to Conjecture~\ref{conj:hyperbola} itself and remains open.

\medskip
\noindent\textbf{(iii) Identification of the missing proof step.}
Define the number of \emph{mixed columns} of $G$ as
\[
    m_{\mathrm{mix}} \;\coloneqq\; n - N(F_1) - N(F_2) 
    \;\ge\; 0,
\]
where the inequality holds because $F_1 \cap F_2 = \{0\}$ implies that no nonzero column can lie in both $F_1$ and $F_2$, so $N(F_1)$ and $N(F_2)$ count disjoint sets of columns. By Theorem~\ref{thm:nonlinear_bound}, the conjecture would follow immediately if $m_{\mathrm{mix}} = 0$, i.e., if every column were pure-file, which is exactly Corollary~\ref{cor:pure_file_hyperbola}. The challenge for general codes is therefore precisely the mixed columns.

A mixed column $g_j \notin F_1 \cup F_2$ has nonzero projections onto both $F_1$ and $F_2$, meaning it is potentially useful to both retrieval processes simultaneously. In the proof of Lemma~\ref{lem:proj_lower}, this simultaneous usefulness causes the bound $\hat{n}_i(\ell) \le n - N(F_{3-i})$ to be loose for both files at the same time: the true number of useful columns in each phase is strictly larger than what the pure-file count accounts for. Mixed columns therefore allow both $E_1$ and $E_2$ to be reduced below the levels predicted by the per-file sequential bound, and the question is whether this joint reduction can be large enough to violate \eqref{eq:hyperbola_conj}. Proving the conjecture requires showing that the gains to $E_1$ and $E_2$ from mixed columns are not independent: any column that is highly useful to $F_1$ recovery must have a large $F_1$-component, which necessarily constrains its $F_2$-component, and vice versa. Making this geometric competition precise in terms of the expected retrieval times is the key open step.

\Cref{fig:nonlinear_n20,fig:nonlinear_n50} illustrate the nonlinear bounds for $k=8$ and $n=20,50$ respectively, across the four partition configurations. The two proved upper bounds (Corollaries~\ref{cor:smax_bound} and~\ref{cor:cs_bound}, solid lines) 
bound the achievability region from below, with Corollary~\ref{cor:cs_bound} 
uniformly tighter. The conjectured hyperbolic boundary $s_1/E_1 + s_2/E_2 = 1$ (Conjecture~\ref{conj:hyperbola}, dashed) lies strictly above both proved bounds for all finite $n$, and the local MDS operating points approach it as $n$ grows, consistent with the asymptotic analysis of Section~\ref{sec:asymptotics}. 

\begin{figure*}[h]
\centering
\includegraphics[width=\linewidth]{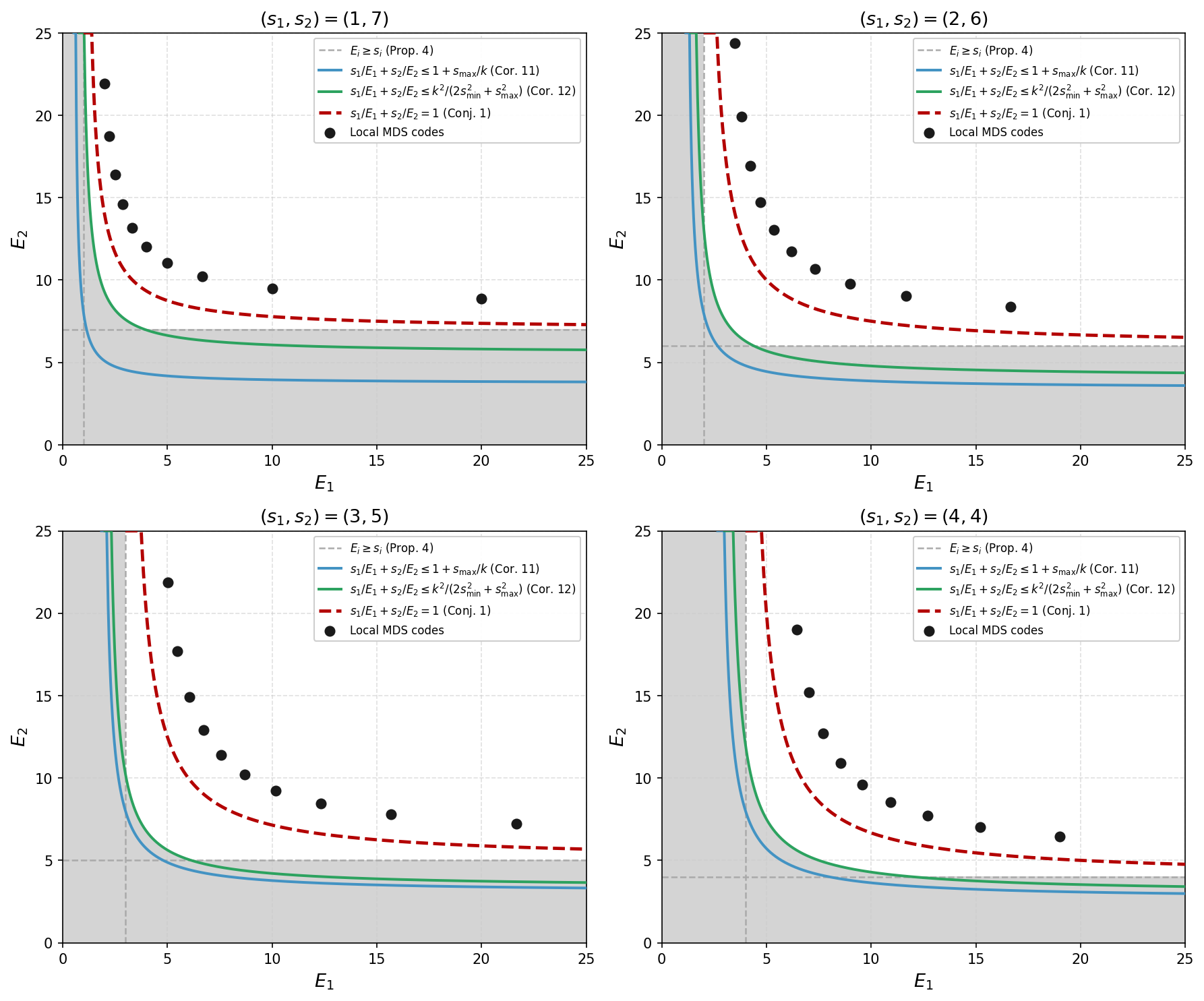}
\caption{Nonlinear bounds for $k=8$, $n=20$, across four partition 
configurations $(s_1,s_2)$. Solid lines are proved bounds 
(Corollaries~\ref{cor:smax_bound} and~\ref{cor:cs_bound}); the dashed 
line is the conjectured hyperbolic boundary 
(Conjecture~\ref{conj:hyperbola}). Black dots are local MDS operating 
points. The gray region is provably non-achievable.}
\label{fig:nonlinear_n20}
\end{figure*}

\begin{figure*}[h]
\centering
\includegraphics[width=\linewidth]{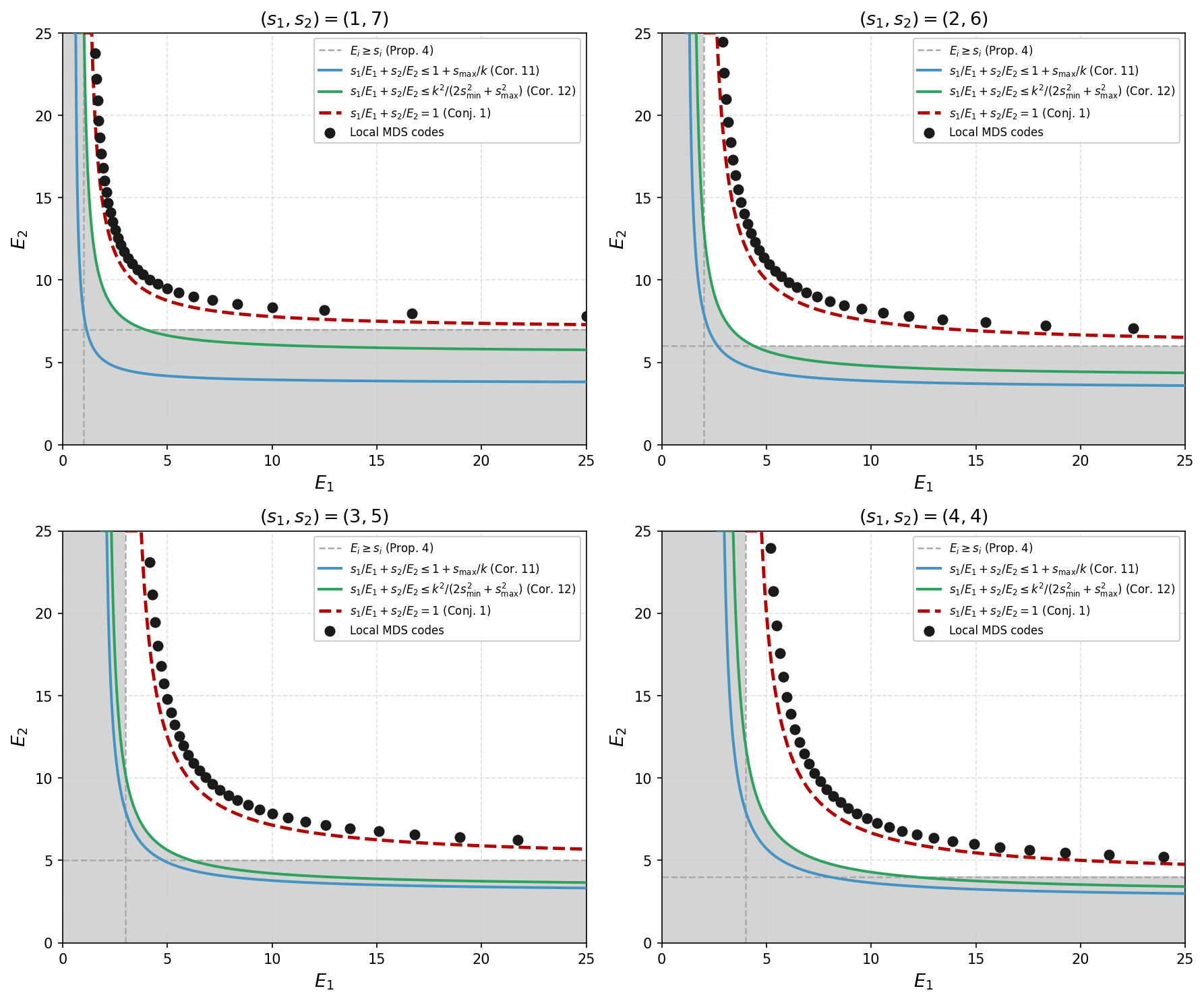}
\caption{Same as \cref{fig:nonlinear_n20} but for $n=50$. The local 
MDS points approach the conjectured hyperbolic boundary as $n$ grows, 
consistent with Theorem~\ref{thm:asymptotic_ub}.}
\label{fig:nonlinear_n50}
\end{figure*}

\begin{remark}\label{rem:s1s2_condition}
The condition $\max\{s_1,s_2\} \ge 2$ in 
Conjecture~\ref{conj:hyperbola} is essential and cannot be removed. 
When $s_1 = s_2 = 1$, equivalently $k = 2$, the bound 
\eqref{eq:hyperbola_conj} can be violated. For example, take $k=2$, 
$n=5$, $s_1=s_2=1$, and the code with generator matrix
\[
    G \;=\; \begin{pmatrix} 1 & 0 & 1 & 0 & 1 \\ 0 & 1 & 0 & 1 & 1 \end{pmatrix},
\]
whose columns are $e_1, e_1, e_2, e_2$, and $e_1+e_2$. As shown in~\cite{gruica2024combinatorial} (Example 2), the expected time to recover 
$F_i = \langle e_i \rangle$, is ${E_i=\frac{23}{12}}$.
Therefore:
\[
    \frac{s_1}{E_1} + \frac{s_2}{E_2} 
    \;=\; \frac{12}{23} + \frac{12}{23} \;=\; \frac{24}{23} \;>\; 1.
\]
The violation occurs because when both files are one-dimensional, a mixed column such as $e_1+e_2$ can contribute to the recovery of $F_1$ indirectly via its linear combination with a subsequently drawn $F_2$-column. This indirect coupling has no analogue when $\max\{s_1,s_2\} \ge 2$, where recovering the larger file requires spanning a subspace of dimension at least $2$ and no single column suffices. Why this distinction is sufficient to restore the hyperbolic constraint in the higher-dimensional case is precisely the content of Conjecture~\ref{conj:hyperbola}.
\end{remark}


\section{Constructions and Finite-$n$ Achievability}
\label{sec:constructions}

We now present explicit code constructions. For each family, we compute the exact expected retrieval times $(E_1, E_2)$ and establish which families dominate others in the sense of~\eqref{eq:pareto}. In each case we also check that the achievable pairs do not violate Conjecture~\ref{conj:hyperbola}, providing further supporting evidence for it.

\subsection{The identity code}\label{ssec:identity}

The identity code $\cC_I$ uses $G_I = I_k$, so that $n = k$ and the columns are exactly the standard basis vectors $e_1, \ldots, e_k$. 
Since file $F_i$ consists of $s_i$ specific basis vectors, recovery of $F_i$ requires drawing all $s_i$ specific columns from the pool $[k]$.

\begin{proposition}\label{prop:identity_Ei}
For the identity code $\cC_I$ with $G_I = I_k$ and any partition $s_1 + s_2 = k$,
\begin{equation}\label{eq:identity_Ei}
    E_i(G_I) \;=\; k\,\harm{s_i}, 
    \qquad i \in \{1,2\}.
\end{equation}
\end{proposition}

\begin{IEEEproof}
Recovery of $F_i$ requires drawing all $s_i$ specific columns from the pool of $k$. Since draws are uniform with replacement, only draws landing on the $s_i$ target columns are useful. When $\ell$ of the $s_i$ target columns have already been collected, the probability of drawing a new one is $(s_i - \ell)/k$, so the expected waiting time in this phase is $k/(s_i - \ell)$. Summing over all $s_i$ phases:
\[
    E_i(G_I) \;=\; \sum_{\ell=0}^{s_i-1} \frac{k}{s_i - \ell}
    \;=\; k\sum_{j=1}^{s_i}\frac{1}{j} \;=\; k\,\harm{s_i}.
    \qedhere
\]
\end{IEEEproof}

\begin{remark}
The recovery of $F_i$ under the identity code is not a standard 
coupon collector problem. In the classical CCP one seeks any $s_i$ out of $k$ coupons; here one must collect a \emph{specific} set of $s_i$ symbols. 
\end{remark}

\begin{corollary}\label{cor:identity_conjecture}
The identity code satisfies Conjecture~\ref{conj:hyperbola} 
whenever $\max\{s_1, s_2\} \ge 2$. 
\end{corollary}

\begin{IEEEproof}
Substituting~\eqref{eq:identity_Ei}:
\[
    \frac{s_1}{E_1(G_I)} + \frac{s_2}{E_2(G_I)} 
    \;=\; \frac{1}{k}\!\left(\frac{s_1}{\harm{s_1}} 
    + \frac{s_2}{\harm{s_2}}\right).
\]
Since $\harm{s} \ge 1$ for all $s \ge 1$, we have $s/\harm{s} \le s$, 
and therefore:
\[
    \frac{s_1}{\harm{s_1}} + \frac{s_2}{\harm{s_2}} 
    \;\le\; s_1 + s_2 \;=\; k,
\]
giving $s_1/E_1 + s_2/E_2 \le 1$. Equality holds if and only 
if $\harm{s_1} = \harm{s_2} = 1$, i.e., $s_1 = s_2 = 1$, 
which requires $k = 2$ and is excluded by the condition 
$\max\{s_1,s_2\} \ge 2$.
\end{IEEEproof}
\subsection{File-dedicated codes}\label{ssec:dedicated}

A file-dedicated code allocates $n_i$ codeword positions exclusively to $F_i$ for $i \in \{1,2\}$, with $n_1 + n_2 = n$ and $n_i \ge s_i$, each protected by an independent linear code. The generator matrix has the block structure
\[
    G(n_1, n_2) \;=\; \begin{pmatrix} G_1 & 0 \\ 0 & G_2 \end{pmatrix},
\]
where $G_1 \in \F_q^{s_1 \times n_1}$ and $G_2 \in \F_q^{s_2 \times n_2}$ are the per-file generator matrices. Since the two blocks are independent, the retrieval of $F_i$ depends only on the column structure of $G_i$. We first show that within this family, MDS codes are optimal for each block.

\begin{lemma}\label{lem:MDS_optimal_block}
For any $[n_i, s_i]$ linear code $G_i$ over $\F_q$,
\[
    E_i(G_i) \;\ge\; E_i(G_i^{\mathrm{MDS}}),
\]
where $G_i^{\mathrm{MDS}}$ is a generator matrix of any $[n_i, s_i]$ MDS code. Equivalently, among all linear codes with $n_i$ columns encoding $s_i$ information symbols, MDS codes minimize the expected retrieval time.
\end{lemma}

\begin{IEEEproof}
By Lemma~\ref{lem:exp_formula_subset}, $E_i$ is a strictly decreasing function of the subset counts $\alpha_{F_i}(s)$. It therefore suffices to show that for any $[n_i, s_i]$ code $G_i$:
\[
    \alpha_{F_i}(s) \;\le\; \binom{n_i}{s} 
    \quad \text{for all } s \ge s_i,
\]
with equality for all $s \ge s_i$ if and only if $G_i$ is MDS.

The upper bound $\alpha_{F_i}(s) \le \binom{n_i}{s}$ is trivial since there are only $\binom{n_i}{s}$ subsets of size $s$ in total. 
For an MDS code, any $s_i$ columns of $G_i$ are linearly independent and span $F_i$, so every $s$-subset with $s \ge s_i$ recovers $F_i$, giving $\alpha_{F_i}^{\mathrm{MDS}}(s) = \binom{n_i}{s}$, which is the maximum possible value. Conversely, if $G_i$ is not MDS, then some set of $s_i$ columns is linearly dependent and fails to span $F_i$, so $\alpha_{F_i}(s_i) < \binom{n_i}{s_i}$, and $E_i(G_i) > E_i(G_i^{\mathrm{MDS}})$.
\end{IEEEproof}

We therefore restrict our attention to file-dedicated codes where both $G_1$ and $G_2$ are MDS codes, and denote by $\cC(n_1,n_2)$ the corresponding code and by $G(n_1,n_2)$ its generator matrix. The valid range of allocations is ${n_1 \in \{s_1, \ldots, n-s_2\}}$, giving $n-k+1$ discrete 
operating points.

\begin{proposition}\label{prop:dedicated_MDS}
Let $i\in\{1,2\}$. For the file-dedicated MDS code $\cC(n_1, n_2)$ with 
$n_1 + n_2 = n$ and $n_i \ge s_i$,
\begin{equation}\label{eq:ded_Ei}
    E_i(G(n_1,n_2)) 
    \;=\; n\sum_{j=1}^{s_i}\frac{1}{n_i - j + 1}
    \;=\; n\bigl(\harm{n_i} - \harm{n_i - s_i}\bigr).
\end{equation}
\end{proposition}

\begin{IEEEproof}
Since all columns with a nonzero $F_i$-component belong to the $G_i$ block, a draw hits a useful column for $F_i$ with probability $n_i/n$ at each step. Because $G_i$ is MDS, any $s_i$ of its $n_i$ columns span $F_i$, so any useful draw that introduces a new column makes progress. When $\ell$ distinct useful columns have been drawn, the probability of drawing a new one is $(n_i - \ell)/n$, giving an expected waiting time of $n/(n_i-\ell)$. 
Summing over all $s_i$ phases:
\begin{align*}
    E_i(G(n_1,n_2)) 
    \;&=\; \sum_{\ell=0}^{s_i-1} \frac{n}{n_i - \ell} 
    \;=\; n\sum_{j=1}^{s_i}\frac{1}{n_i - j + 1} \\
    \;&=\; n\Bigl(\frac{1}{n_i} + \frac{1}{n_i-1} 
    + \cdots + \frac{1}{n_i-s_i+1}\Bigr)\\ 
    \;&=\; n\bigl(\harm{n_i} - \harm{n_i-s_i}\bigr). \qedhere
\end{align*}
\end{IEEEproof}

As $n_1$ varies from $s_1$ to $n-s_2$, the pair $(E_1, E_2)$ traces $n-k+1$ discrete operating points, with $E_1$ decreasing monotonically and $E_2$ increasing monotonically in $n_1$. No two of these points are Pareto-comparable.

\begin{corollary}\label{cor:dedicated_conjecture}
Every file-dedicated code satisfies Conjecture~\ref{conj:hyperbola}.
\end{corollary}

\begin{IEEEproof}
By Lemma~\ref{lem:MDS_optimal_block}, for any file-dedicated code 
with per-file codes $G_1, G_2$, we have $E_i(G_i) \ge E_i(G_i^{\mathrm{MDS}})$ for both $i$, and hence $s_i/E_i(G_i) \le s_i/E_i(G_i^{\mathrm{MDS}})$. Since every column of any file-dedicated code belongs entirely to either $F_1$ or $F_2$, we have $N(F_1) + N(F_2) = n$, and Corollary~\ref{cor:pure_file_hyperbola} gives $s_1/E_1^{\mathrm{MDS}} + s_2/E_2^{\mathrm{MDS}} \le 1$. Combining both inequalities completes the proof.
\end{IEEEproof}

\begin{corollary}\label{cor:identity_dominated}
The identity code $\cC_I$ coincides with the minimum file-dedicated allocation $\cC(s_1, s_2)$ at $n = k$, and is strictly Pareto-dominated by $\cC(n_1, n_2)$ for any allocation with $n > k$.
\end{corollary}

\begin{IEEEproof}
At $n = k$, the only valid allocation is $n_1 = s_1$, $n_2 = s_2$, 
and $\cC(s_1, s_2) = \cC_I$, so the two codes are identical. 
For $n > k$, increasing $n_i$ beyond $s_i$ strictly increases 
$\alpha_{F_i}(s)$ for all $s \ge s_i$ and hence strictly reduces 
$E_i$ for both $i$ simultaneously.
\end{IEEEproof}

\subsection{Systematic $[n,k]$ MDS codes}\label{ssec:MDS_construction}

Let $\cC^n_\text{Global}$ be a systematic $[n,k]$ MDS code over $\F_q$ with generator matrix $G^n_\text{Global} = [I_k \mid P]$, where $P \in \F_q^{k \times (n-k)}$ is the parity block. The $k$ systematic columns are the standard basis vectors $e_1,\ldots,e_k$, and any $k$ columns of $G^n_\text{Global}$ are linearly independent.

\begin{proposition}\label{prop:MDS_alpha}
For a systematic $[n,k]$ MDS code with target file $F_i$ of dimension $s_i$, the subset counts are:
\begin{equation}\label{eq:MDS_alpha}
    \alpha_{F_i}(s) \;=\; 
    \begin{cases}
        0 & s < s_i, \\
        \displaystyle\binom{n-s_i}{s-s_i} & s_i \le s \le k-1, \\
        \displaystyle\binom{n}{s} & s \ge k.
    \end{cases}
\end{equation}
\end{proposition}

\begin{IEEEproof}
For $s < s_i$ no $s$-subset can span $F_i$, giving zero. For $s \ge k$, any $k$ columns are linearly independent and span $\F_q^k \supseteq F_i$, so all $\binom{n}{s}$ subsets recover $F_i$.

For $s_i \le s \le k-1$, we claim a subset $S$ of size $s$ spans $F_i$ if and only if $S$ contains all $s_i$ systematic columns of $F_i$. For the $(\Leftarrow)$ direction, if $S \supseteq \{e_1,\ldots,e_{s_i}\}$ then ${F_i \subseteq \langle g_j : j \in S\rangle}$ trivially. For the $(\Rightarrow)$ direction, suppose $S$ is missing some systematic column $e_j \in F_i$. Then $S \cup \{e_j\}$ has size $s+1 \le k$, so by the MDS property these columns are linearly independent, which means $e_j \notin \langle g_j : j \in S\rangle$, and therefore $F_i \not\subseteq \langle g_j : j \in S\rangle$.

Given that the $s_i$ systematic columns of $F_i$ must all be included, the remaining $s-s_i$ columns may be chosen freely from the $n-s_i$ remaining columns, giving 
$\alpha_{F_i}(s) = \binom{n-s_i}{s-s_i}$.
\end{IEEEproof}

\begin{proposition}\label{prop:MDS_Ei}
For a systematic $[n,k]$ MDS code and any partition $s_1 + s_2 = k$, the expected retrieval time of file $F_i$ is:
\begin{equation}\label{eq:MDS_Ei}
    E_i(G^n_\text{Global}) \;=\; n(H_n - H_{n-k}) 
    - \frac{n}{\binom{n}{s_i}}
    \sum_{s=s_i}^{k-1}\frac{\binom{s}{s_i}}{n-s}.
\end{equation}
\end{proposition}

\begin{IEEEproof}
Substituting~\eqref{eq:MDS_alpha} into~\eqref{eq:file_exp}:
\begin{align*}
    E_i(G) \;&=\; nH_n - \sum_{s=s_i}^{k-1}
    \frac{\binom{n-s_i}{s-s_i}}{\binom{n-1}{s}}
    - \sum_{s=k}^{n-1}\frac{\binom{n}{s}}{\binom{n-1}{s}}.
\end{align*}
For the second sum, $\binom{n}{s}/\binom{n-1}{s} = n/(n-s)$, so $$\sum_{s=k}^{n-1}\frac{n}{n-s} = nH_{n-k}.$$ For the first sum, using $\binom{n-s_i}{s-s_i}/\binom{n}{s} = \binom{s}{s_i}/\binom{n}{s_i}$ and $\binom{n}{s}/\binom{n-1}{s} = n/(n-s)$ 
gives 
$$\frac{\binom{n-s_i}{s-s_i}}{\binom{n-1}{s}} = 
\frac{n}{\binom{n}{s_i}}\cdot\frac{\binom{s}{s_i}}{n-s}.$$
Substituting yields~\eqref{eq:MDS_Ei}.
\end{IEEEproof}

\begin{corollary}[$s_1 = 1$ case]\label{cor:MDS_s1_one}
For $s_1 = 1$ and any $n \ge k$, the expected retrieval time of the single-symbol file satisfies $E_1(G^n_\text{Global}) = k$, independently of $n$.
\end{corollary}

\begin{IEEEproof}
Setting $s_i = 1$ in~\eqref{eq:MDS_Ei} gives
\[
    E_1(G^n_\text{Global}) \;=\; n(H_n - H_{n-k}) 
    - \sum_{s=1}^{k-1}\frac{s}{n-s}.
\]
Substituting $j = n-s$:
\[
    \sum_{s=1}^{k-1}\frac{s}{n-s} 
    \;=\; \sum_{j=n-k+1}^{n-1}\frac{n-j}{j}
    \;=\; \sum_{j=n-k+1}^{n-1}\frac{n}{j} - (k-1).
\]
Since 
$$n(H_n - H_{n-k}) = \sum_{j=n-k+1}^{n} \frac{n}{j} = 
\sum_{j=n-k+1}^{n-1}\frac{n}{j} + 1,$$ we get:
\[
    E_1(G^n_\text{Global}) \;=\; \sum_{j=n-k+1}^{n-1}\frac{n}{j} + 1 
    - \sum_{j=n-k+1}^{n-1}\frac{n}{j} + (k-1) \;=\; k.
    \qedhere
\]
\end{IEEEproof}

We note that the result of Corollary~\ref{cor:MDS_s1_one} was proven independently in~\cite{bar-lev2025cover,gruica2024combinatorial}  as part of the analysis of the random access case where the objective is the retrieval of a single information symbol.

\begin{corollary}[Symmetric case]\label{cor:MDS_symmetric}
For $s_1 = s_2 = k/2$, the expected retrieval time of each file under a systematic $[n,k]$ MDS code is:
\begin{equation}\label{eq:MDS_symmetric}
    E_i(G^n_\text{Global}) \;=\; n(H_n - H_{n-k}) 
    - \frac{n}{\binom{n}{k/2}}
    \sum_{s=k/2}^{k-1}\frac{\binom{s}{k/2}}{n-s}.
\end{equation}
\end{corollary}

We now show that there is always a local MDS code that dominates the global systematic MDS code. First, we need the following two lemmas, which are proved in the Appendix.

\begin{lemma}\label{lem:hyp_representation}
Let $f(x) \coloneqq \harm{s_i} - \harm{s_i - x}$ for $x \in \{0,\ldots,s_i\}$.
Choose $k$ of the $n$ columns of $G^n_\mathrm{Global}$ uniformly at random, and let $X_G$ be the number of systematic columns of $F_i$ among the chosen columns. Similarly, choose $s_i$ of the $n_i$ columns of block $i$ of $G(n_1,n_2)$ uniformly at random, and let $X_L$ be the number of the first $s_i$ columns of block $i$ among the chosen columns. That is,
\[
    \Pr[X_G = x] = \frac{\binom{s_i}{x}\binom{n-s_i}{k-x}}{\binom{n}{k}},
    \qquad
    \Pr[X_L = x] = \frac{\binom{s_i}{x}\binom{n_i-s_i}{s_i-x}}{\binom{n_i}{s_i}},
\]
for $x \in \{0,\ldots,s_i\}$. Then
\begin{equation}\label{eq:hyp_representation}
    E_i(G^n_\mathrm{Global}) = n\,\E[f(X_G)]
    \qquad \text{and} \qquad
    E_i(G(n_1,n_2)) = n\,\E[f(X_L)].
\end{equation}
\end{lemma}

Note that $n f(x) = \sum_{r=1}^{x}\frac{n}{s_i-r+1}$ is the expected number of uniform draws needed to collect $x$ out of $s_i$ designated columns: when $r-1$ of them have been collected, the next one is drawn with probability $(s_i-r+1)/n$ per step. Thus Lemma~\ref{lem:hyp_representation} states that each code's expected retrieval time equals the expected time to collect as many designated columns as a uniformly random subset of the appropriate size contains.

\begin{lemma}\label{lem:hyp_convex_order}
Let $k \mid n$ and $n_i = \frac{ns_i}{k}$. Then
$\E[\varphi(X_L)] \le \E[\varphi(X_G)]$ for every convex function $\varphi$ on $\{0,\ldots,s_i\}$.
\end{lemma}

\begin{proposition}\label{prop:local_dominates_global}
Let $k \mid n$. For any partition $s_1 + s_2 = k$ and the proportional
allocation $n_i = \frac{ns_i}{k}$, the local MDS code satisfies
\[
    E_i(G(n_1,n_2)) \;\le\; E_i(G^n_\mathrm{Global})
    \quad \text{for both } i \in \{1,2\}.
\]
\end{proposition}

\begin{IEEEproof}
The function $f(x) = \harm{s_i} - \harm{s_i-x}$ is convex on
$\{0,\ldots,s_i\}$, since its increments $f(x+1)-f(x) = \frac{1}{s_i-x}$,
$x \in \{0,\ldots,s_i-1\}$, are increasing in $x$. Combining
Lemma~\ref{lem:hyp_representation} with Lemma~\ref{lem:hyp_convex_order}
applied to $\varphi = f$ completes the proof.
\end{IEEEproof}

\begin{remark}
For the symmetric case $s_1 = s_2 = k/2$, the proportional allocation reduces
to $n_1 = n_2 = n/2$, and Proposition~\ref{prop:local_dominates_global}
provides a significantly simpler proof of~\cite[Lemma~2]{abraham2024covering}
for the two-file case. The Abraham et al.\ proof proceeds via an explicit
Markov chain construction with transition matrices and absorbing state
analysis. Our proof expresses both expected retrieval times as the mean of a
single convex function of a hypergeometric random variable, with the
proportional allocation being exactly the choice that equalizes the two
means.
\end{remark}

\begin{corollary}\label{cor:global_MDS_conjecture}
For $k \mid n$, the systematic global $[n,k]$ MDS code satisfies 
Conjecture~\ref{conj:hyperbola} for any partition $s_1 + s_2 = k$.
\end{corollary}

\begin{IEEEproof}
By Proposition~\ref{prop:local_dominates_global}, 
$E_i(G^n_\text{Global}) \ge E_i(G(n_1,n_2))$ for both $i$, 
hence $s_i/E_i(G^n_\text{Global}) \le s_i/E_i(G(n_1,n_2))$. 
Since $G(n_1,n_2)$ satisfies the conjecture by 
Corollary~\ref{cor:dedicated_conjecture}, the result follows.
\end{IEEEproof}

\subsection{Beyond File-Dedicated Codes}\label{ssec:hybrid}

The file-dedicated family of Section~\ref{ssec:dedicated} provides $n-k+1$ achievable operating points for fixed $n$, no two of which are Pareto-comparable. A natural question is whether the achievability region $\Lambda_{n,k}(s_1,s_2)$ contains points that are not achievable by any file-dedicated code.

Example~\ref{ex:hybrid_n8} answers this affirmatively for specific parameters: for $k=4$, $s_1=1$, $s_2=3$, $n=8$, the generator matrix ${G}^*$ achieves $E_1 \approx 3.84$, a value strictly between Cases~B and~C of the dedicated family, and is not dominated by either. This shows that for these parameters, the achievability region contains points not achievable by any file-dedicated code.

Whether ${G}^*$ is itself dominated by some other non-dedicated generator matrix remains open. We note that ${G}^*$ satisfies Conjecture~\ref{conj:hyperbola}, but sits further from the hyperbolic boundary than the dedicated codes, consistent with the intuition that mixed columns introduce a cost in terms of 
the overall trade-off efficiency.

A richer picture emerges when we allow the code length to grow. Given any two codes $G^A$ and $G^B$ with the same parameters $(n, k)$ achieving points $A = (E_1^A, E_2^A)$ and $B = (E_1^B, E_2^B)$ respectively, consider the block concatenation of $\lambda$ copies of $G^A$ and $\mu$ copies of $G^B$, yielding a $k \times (\lambda+\mu)n$ code. As $\lambda/\mu$ varies, the achievable point lies somewhere 
between $A$ and $B$, though its exact location depends on the 
interplay between the two codes and is difficult to characterize. This observation connects naturally to the asymptotic analysis of Section~\ref{sec:asymptotics}, where we characterize the achievability region as $n \to \infty$.

\section{Asymptotic Analysis of the Achievability Region}
\label{sec:asymptotics}

\subsection{Monotonicity in $n$}

We first establish that the achievability region can only grow as $n$ increases by multiples. 

\begin{proposition}\label{prop:monotone_mn}
For any $m \ge 1$ and any $(E_1, E_2) \in \Lambda_{n,k}(s_1,s_2)$, there exists a generator matrix of length $mn$ achieving the same pair $(E_1, E_2)$. Consequently, $\Lambda_{n,k}(s_1,s_2) \subseteq \Lambda_{mn,k}(s_1,s_2)$.
\end{proposition}

\begin{IEEEproof}
Given a generator matrix $G \in \F_q^{k \times n}$ achieving $(E_1, E_2)$, consider the repeated generator matrix $G^{(m)} = [G | G | \cdots | G] \in \F_q^{k \times mn}$ consisting of $m$ identical copies of $G$. 
Each column $g_j$ of $G$ appears $m$ times in $G^{(m)}$, so under uniform sampling from $mn$ columns, each column type $g_j$ is drawn with total probability $m/(mn) = 1/n$, identical to sampling from $G$. Since recovery of $F_i$ depends only on which column types are drawn and not which physical copy, the entire retrieval time distribution of $G^{(m)}$ is identical to that of $G$, giving $E_i(G^{(m)}) = E_i(G)$ for both $i$.
\end{IEEEproof}

\begin{remark}
Proposition~\ref{prop:monotone_mn} shows that $\Lambda_{n,k} \subseteq \Lambda_{mn,k}$ for any integer $m \ge 1$. The stronger statement that $\Lambda_{n,k} \subseteq \Lambda_{n+1,k}$ is geometrically natural and consistent with all our examples, but establishing it requires a more careful analysis and is left for future work. We note that the stronger statement does not hold for $s_1=s_2=1$ (see Remark~\ref{rem:monotone_counterexample}), but we conjecture it holds once $\max\{s_1,s_2\} \ge 2$
\end{remark}

\begin{conjecture}\label{conj:monotone}
For any $n \ge k$, and any partition $s_1 + s_2 = k$ with $\max\{s_1,s_2\} \ge 2$,  $\Lambda_{n,k}(s_1,s_2) \subseteq \Lambda_{n+1,k}(s_1,s_2)$ in the sense that for every $(E_1,E_2) \in \Lambda_{n,k}(s_1,s_2)$ there exists a a rank-$k$ matrix with $n+1$ columns achieving $(E_1',E_2')$ with $E_i' \le E_i$ for both $i \in \{1,2\}$.
\end{conjecture}

\begin{remark}\label{rem:monotone_counterexample}
The condition $\max\{s_1,s_2\} \ge 2$ in Conjecture~\ref{conj:monotone} is essential and cannot be removed. When $s_1 = s_2 = 1$, equivalently $k = 2$, the inclusion can fail. For example, take $s_1 = s_2 = 1$, $n = 3$, and the code with generator matrix
\[
    G \;=\; \begin{pmatrix} 1 & 1 & 0 \\ 0 & 0 & 1 \end{pmatrix},
\]
whose columns are $e_1, e_1, e_2$. It achieves $(E_1, E_2) = \bigl(\tfrac{3}{2}, 3\bigr)$, and no rank-$2$ matrix with $4$ columns achieves both $E_1 \le \tfrac{3}{2}$ and $E_2 \le 3$, so $\Lambda_{3,2}(1,1) \not\subseteq\Lambda_{4,2}(1,1)$. The failure relies on both files being one-dimensional: a single column places $F_1$
within reach, so concentrating $n-1$ columns on the $e_1$-line makes $E_1$ near-minimal
while forcing $E_2 = n$, and no length-$(n+1)$ matrix can improve both. This has no analog when $\max\{s_1,s_2\} \ge 2$, since recovering the larger file requires spanning a subspace of dimension at least $2$, which no single column achieves. 
\end{remark}

Given two achievable points at the same length $n$, one can form their block concatenation: if $G^A, G^B \in \F_q^{k \times n}$ achieve $(E_1^A, E_2^A)$ and $(E_1^B, E_2^B)$ respectively, then $[G^A | G^B] \in \F_q^{k \times 2n}$ is a valid code of length $2n$. Recovery of $F_i$ from $[\,G^A \mid G^B\,]$ succeeds whenever the drawn columns span $F_i$, which can happen using columns from $G^A$ alone, from $G^B$ alone, or from a combination of both. In particular, using only the $G^A$ columns, recovery of $F_i$ requires drawing $n_A$-type columns, each of which arrives with probability $n/(2n) = 1/2$, giving an upper bound $E_i([\,G^A \mid G^B\,]) \le 2E_i^A$, and symmetrically $\le 2E_i^B$. The concatenated code therefore achieves a point no worse than $2\min(E_i^A, E_i^B)$ in each coordinate, but its exact location in the achievability region depends on the interplay between the two codes and is difficult to characterize in general. Whether the concatenation always achieves a point strictly between $A$ and $B$ in the Pareto sense remains open.

\subsection{Asymptotic Achievability}
\begin{theorem}\label{thm:asymptotic_ub}
Fix any sequence of allocations $(n_1^{(n)}, n_2^{(n)})$ with $n_1^{(n)} + n_2^{(n)} = n$, $n_i^{(n)} \ge s_i$, and $n_i^{(n)}/n \to \rho_i$ as $n \to \infty$ for some $\rho_1, \rho_2 > 0$ with $\rho_1 + \rho_2 = 1$. Then the file-dedicated MDS code with this allocation satisfies:
\begin{equation}\label{eq:asymptotic_Ei}
    E_i(G(n_1^{(n)}, n_2^{(n)})) \;\to\; \frac{s_i}{\rho_i} 
    \quad \text{as } n \to \infty,
\end{equation}
and consequently $s_1/E_1 + s_2/E_2 \to 1$. Therefore, for any target point $(E_1^*, E_2^*)$ with $s_1/E_1^* + s_2/E_2^* = 1$ and $E_i^* > s_i$, the allocation $n_i^{(n)} = \lfloor ns_i/E_i^* \rfloor + \epsilon_i^{(n)}$, where $\epsilon_i^{(n)} \in \{0,1\}$ is chosen to ensure $n_1^{(n)} + n_2^{(n)} = n$, achieves $(E_1, E_2) \to (E_1^*, E_2^*)$.
\end{theorem}

\begin{IEEEproof}
By Proposition~\ref{prop:dedicated_MDS}, 
$$E_i(G(n_1^{(n)},n_2^{(n)})) = n(\harm{n_i^{(n)}} - 
\harm{n_i^{(n)}-s_i}).$$ Writing this as a sum:
\[
    E_i(G(n_1^{(n)},n_2^{(n)})) 
    \;=\; \sum_{j=1}^{s_i} \frac{n}{n_i^{(n)} - j + 1}.
\]
Since $n_i^{(n)}/n \to \rho_i$, for each fixed 
$j \in \{1,\ldots,s_i\}$:
\[
    \frac{n}{n_i^{(n)} - j + 1} 
    \;=\; \frac{1}{n_i^{(n)}/n - (j-1)/n} 
    \;\to\; \frac{1}{\rho_i},
\]
since $(j-1)/n \to 0$. Summing over $s_i$ terms each converging to $1/\rho_i$ gives $E_i \to s_i/\rho_i$. For any target point $(E_1^*, E_2^*)$ with $s_1/E_1^* + s_2/E_2^* = 1$ and $E_i^* > s_i$, set $\rho_i = s_i/E_i^*$, so that $\rho_1 + \rho_2 = 1$. Define $n_1^{(n)} = \lfloor n\rho_1 \rfloor$ and $n_2^{(n)} = n - n_1^{(n)}$, so that $n_1^{(n)} + n_2^{(n)} = n$ exactly and $n_i^{(n)}/n \to \rho_i$ as $n \to \infty$. By the first part of the theorem, $E_i \to s_i/\rho_i = E_i^*$.
\end{IEEEproof}

\begin{corollary}\label{cor:boundary_dense}
The hyperbolic curve $\{(E_1,E_2) : s_1/E_1 + s_2/E_2 = 1,\, E_i \ge s_i\}$ is the closure of the set of asymptotically achievable points via file-dedicated MDS codes. In particular, Conjecture~\ref{conj:hyperbola} is asymptotically tight and the bound $\le 1$ cannot be replaced by any strictly smaller constant without conditioning on specific values of $n$.
\end{corollary}

\subsection{The Limiting Achievability Region}

Theorem~\ref{thm:asymptotic_ub} shows that the full hyperbolic boundary $\{s_1/E_1 + s_2/E_2 = 1,\, E_i \ge s_i\}$ is 
asymptotically achievable by file-dedicated MDS codes. Together 
with Conjecture~\ref{conj:hyperbola}, this gives a complete 
conjectured characterization of the limiting achievability region.

\begin{conjecture}\label{conj:limiting}
Denoting by $\mathrm{cl}(\cdot)$ the topological closure in 
$\mathbb{R}^2$, for any $s_1, s_2$ with $\max\{s_1,s_2\} \ge 2$:
\[
    \mathrm{cl}\!\left(\bigcup_{n \ge k} 
    \Lambda_{n,k}(s_1,s_2)\right)
    \;=\;
    \left\{(E_1,E_2) \in \mathbb{R}_{>0}^2 : \substack{
        \frac{s_1}{E_1} + \frac{s_2}{E_2} \le 1,\;\\
        E_i \ge s_i}
    \right\}.
\]
\end{conjecture}

The inclusion $\supseteq$ is established by Theorem~\ref{thm:asymptotic_ub} where we show that the full hyperbolic boundary is asymptotically achievable by file-dedicated MDS codes. The inclusion $\subseteq$, i.e., no sequence of codes can approach a point strictly outside the hyperbolic region, is equivalent to Conjecture~\ref{conj:hyperbola} applied uniformly over all $n$.


\section{Conclusion}\label{sec:conclusion}

We have introduced and analyzed the block-structured coded retrieval problem for two files, characterizing the achievability region $\Lambda_{n,k}(s_1,s_2)$ of expected retrieval time pairs $(E_1, E_2)$ over all linear codes of length $n$.

Our main contributions are as follows. We derived a family of linear lower bounds via mutual exclusivity of recovery sets, culminating in the tightest polytope cut of Corollary~\ref{cor:tightest_cut}. We then developed a nonlinear bound via column projection (Theorem~\ref{thm:nonlinear_bound}), which for codes with no mixed columns yields the hyperbolic constraint $s_1/E_1 + s_2/E_2 \le 1$ (Corollary~\ref{cor:pure_file_hyperbola}). We computed exact expected retrieval times for the identity code, file-dedicated MDS codes, and systematic global MDS codes and showed that all satisfy the hyperbolic constraint. For the global MDS code, we established dominance by the proportional local MDS allocation (Proposition~\ref{prop:local_dominates_global}), with a proof via convex ordering of hypergeometric distributions that is significantly simpler than, and extends, the argument of~\cite{abraham2024covering}.
 Finally, we characterized the limiting achievability region as $n \to \infty$: the hyperbolic boundary is asymptotically achieved by file-dedicated MDS codes (Theorem~\ref{thm:asymptotic_ub}), and the full limiting region is conjectured to equal the hyperbolic region (Conjecture~\ref{conj:limiting}).

The central open problem is Conjecture~\ref{conj:hyperbola}: 
that $s_1/E_1 + s_2/E_2 \le 1$ holds for every linear code whenever $\max\{s_1,s_2\} \ge 2$. We have established this for all pure-file codes and for global MDS codes, and verified it numerically for a large sample of random generator matrices, but a general proof remains elusive. The key challenge is bounding the joint contribution of mixed columns to both retrieval processes simultaneously, as discussed in Section~\ref{sec:nonlinear_bound}.

Several further directions remain open.

\smallskip
\noindent\textbf{Monotonicity in $n$.} 
We proved that $\Lambda_{n,k} \subseteq \Lambda_{mn,k}$ for any integer $m \ge 1$ via column duplication (Proposition~\ref{prop:monotone_mn}). The stronger statement $\Lambda_{n,k} \subseteq \Lambda_{n+1,k}$ for $\max\{s_1,s_2\}\ge 2$ is geometrically natural and consistent with all our examples, but requires more careful analysis and is left for future work (Conjecture~\ref{conj:monotone}).

\smallskip
\noindent\textbf{Non-divisible code lengths.} 
The dominance of proportional local MDS over global MDS (Proposition~\ref{prop:local_dominates_global}) requires $k \mid n$. For general $n$, the relationship between the floor/ceiling allocations $n_i \in \{\lfloor ns_i/k \rfloor, \lceil ns_i/k \rceil\}$ and the global MDS code requires careful analysis and is left for future work.

\smallskip
\noindent\textbf{Beyond file-dedicated codes.}  Example~\ref{ex:hybrid_n8} shows that for specific parameters, codes with mixed columns can achieve points in $\Lambda_{n,k}(s_1,s_2)$ not achievable by any file-dedicated code. Whether such points are dominated by other non-dedicated codes, and what the full Pareto boundary looks like for fixed $n$, remain open. The concatenation of two codes of the same length provides achievable points at length $2n$, but characterizing the exact location of these points in the achievability region is difficult in general.

\smallskip
\noindent\textbf{The multi-file case.} The natural generalization to $f > 2$ files introduces an achievability region in $\mathbb{R}^f$. The file-dedicated construction generalizes directly, providing $f$-dimensional operating points, but the structure of the achievability region, the appropriate generalization of the hyperbolic bound, and the analogue of Conjecture~\ref{conj:hyperbola} for $f > 2$ files all remain open.

\section*{Acknowledgment}
The author thanks the Bolzano.app project at Charles University for a helpful observation that led to Remark~\ref{rem:monotone_counterexample}.


\bibliographystyle{IEEEtran}
\bibliography{refs.bib}

\newpage
\begin{appendix}
\textit{Proof of Lemma~\ref{lem:hyp_representation}}:
First, we rewrite $E_i(G^n_\mathrm{Global})$ as a single sum. Since
$\binom{s}{s_i} = 0$ for $s < s_i$ and
$\harm{n} - \harm{n-k} = \sum_{s=0}^{k-1}\frac{1}{n-s}$,
Proposition~\ref{prop:MDS_Ei} gives
\begin{equation}\label{eq:glo_regrouped}
    E_i(G^n_\mathrm{Global})
    = n\sum_{s=0}^{k-1}\frac{1}{n-s}
      \left(1 - \frac{\binom{s}{s_i}}{\binom{n}{s_i}}\right).
\end{equation}
Next, we convert this sum into an expectation over $X_G$. By Vandermonde's identity, $$\binom{n}{s_i} - \binom{s}{s_i}
= \sum_{j=0}^{s_i}\binom{s}{s_i-j}\binom{n-s}{j} - \binom{s}{s_i-0}\binom{n-s}{0}=\sum_{j=1}^{s_i}\binom{s}{s_i-j}\binom{n-s}{j},$$ and
$\frac{1}{n-s}\binom{n-s}{j} = \frac{1}{j}\binom{n-s-1}{j-1}$. Substituting
both into~\eqref{eq:glo_regrouped} yields
\begin{equation}\label{eq:after_vandermonde}
    E_i(G^n_\mathrm{Global})
    = n\sum_{j=1}^{s_i}\frac{1}{j}\sum_{s=0}^{k-1}
      \frac{\binom{s}{s_i-j}\binom{n-s-1}{j-1}} {\binom{n}{s_i}}.
\end{equation}
To interpret the inner sum, arrange the $n$ columns in uniformly random order; then $X_G$ is distributed as the number of systematic columns of $F_i$ among the first $k$ positions, since the first $k$ positions form a uniformly
random $k$-subset. The summand in~\eqref{eq:after_vandermonde} is the probability that position $s+1$ holds a systematic column of $F_i$ with exactly $s_i-j$ such columns before it, i.e., that the $(s_i-j+1)$-th systematic column of $F_i$ occupies position $s+1$. Indeed, out of the $\binom{n}{s_i}$ equally likely placements of the systematic columns of $F_i$ among the $n$ positions, the placements realizing this event are counted by choosing $s_i-j$ of them among the first $s$ positions and the remaining $j-1$ among the last $n-s-1$ positions.
Summing over $s = 0,\ldots,k-1$, the inner sum equals $\Pr[X_G \ge s_i-j+1]$. Substituting $r = s_i-j+1$
into~\eqref{eq:after_vandermonde} gives
\[
    E_i(G^n_\mathrm{Global})
    = n\sum_{r=1}^{s_i}\frac{\Pr[X_G \ge r]}{s_i-r+1}
    = n\,\E\!\left[\sum_{r=1}^{X_G}\frac{1}{s_i-r+1}\right]
    = n\,\E[f(X_G)].
\]

Lastly, we treat the local code. By Proposition~\ref{prop:dedicated_MDS},
$E_i(G(n_1,n_2)) = n\sum_{s=0}^{s_i-1}\frac{1}{n_i-s}$. Since
$\binom{s}{s_i} = 0$ for every $s \le s_i - 1$, each summand may be multiplied by $1 - \binom{s}{s_i}/\binom{n_i}{s_i} = 1$, giving
\[
    E_i(G(n_1,n_2))
    = n\sum_{s=0}^{s_i-1}\frac{1}{n_i-s}
      \left(1 - \frac{\binom{s}{s_i}}{\binom{n_i}{s_i}}\right).
\]
The sum has the same form as the sum in~\eqref{eq:glo_regrouped}, with $(n_i, s_i)$ in place of $(n, k)$; note that the prefactor $n$ is unchanged, as it plays no role in the computation. Repeating the computation above with these parameters yields
$E_i(G(n_1,n_2)) = n\,\E[f(X_L)]$.
$\blacksquare$

\textit{Proof of Lemma~\ref{lem:hyp_convex_order}}: 
First, the two variables have the same mean. Writing a hypergeometric random variable as a sum of indicators, one per drawn element, each indicator has an expectation equal to the fraction of designated elements in the population,
hence
\[
    \E[X_G] = k\cdot\frac{s_i}{n} = \frac{ks_i}{n},
    \qquad
    \E[X_L] = s_i\cdot\frac{s_i}{n_i} = \frac{ks_i}{n},
\]
where the last equality holds precisely because $n_i = \frac{ns_i}{k}$.

If $n = k$, then $n_i = s_i$, so in both experiments the chosen subset is the entire population and $X_G = X_L = s_i$ deterministically; the claim then holds with equality. Assume therefore that $n \ge 2k$; then both variables are supported on all of
$\{0,\ldots,s_i\}$.

Next, we analyze the sign pattern of the difference between the two probability mass functions,
\[
    d(x) \;\coloneqq\; \Pr[X_G = x] - \Pr[X_L = x],
    \qquad x \in \{0, \ldots, s_i\}.
\]
By the symmetry of the hypergeometric distribution, we may write
$\Pr[X_G = x] = \binom{k}{x}\binom{n-k}{s_i-x}/\binom{n}{s_i}$; for $X_L$ the
two symmetric forms coincide. The increments of the log-likelihood ratio
$\ell(x) \coloneqq \log\Pr[X_G = x] - \log\Pr[X_L = x]$ are then, for
$x \in \{0, \ldots, s_i-1\}$,
\[
    \ell(x+1) - \ell(x)
    = \log\frac{k-x}{s_i-x}
    + \log\frac{(n_i-s_i)-(s_i-x)+1}{(n-k)-(s_i-x)+1}.
\]
The first term is increasing in $x$, since $\frac{k-x}{s_i-x} = 1 + \frac{k-s_i}{s_i-x}$ and $s_i < k$. The second term is also increasing in $x$: writing it as $\log\frac{A-y}{B-y}$ with $A = n_i - s_i$, $B = n - k$, and $y = s_i - x - 1$, the ratio is decreasing in $y$ since $A < B$, and $y$ is decreasing in $x$. Hence the increments of $\ell$ are nondecreasing.

We claim that, as $x$ runs from $0$ to $s_i$, the sign sequence of $d$ consists of nonnegative values, followed by negative values, followed again by nonnegative values, where the middle block is nonempty and at least one value in each outer block is strictly positive; we abbreviate this pattern by $+,-,+$. First, since $\ell(x) < 0$ exactly when $d(x) < 0$, the set $\{x : d(x) < 0\}$ is a set of consecutive integers: if $\ell$ were negative at two points and nonnegative at some point between them, then some increment of $\ell$ would be positive and a later one negative, contradicting the monotonicity of the increments. Therefore, the negative values of $d$ form a single block. Second, no block can cover the whole range, and the blocks cannot vanish in a way that leaves at most one sign change. Indeed, since both probability mass functions sum to $1$, we have $\sum_{x} d(x) = 0$, so $d$ must take both signs unless it vanishes identically. Moreover, if $d$ had exactly one sign change, say nonnegative up to some point and negative beyond, then the partial sums $\sum_{y \le x} d(y) = F_{X_G}(x) - F_{X_L}(x)$ would be nonnegative for all $x$, that is, $F_{X_G} \ge F_{X_L}$ everywhere. Since $\E[X] = \sum_{x=0}^{s_i-1}(1 - F_X(x))$ for any random variable on $\{0,\ldots,s_i\}$, this would force $\E[X_G] \le \E[X_L]$ with equality only if the two distributions coincide, and similarly for the mirrored pattern. As the means are equal, either the distributions coincide, in which case the claim of the lemma is immediate, or the sign sequence of $d$ is $+,-,+$, that is, $d$ changes sign exactly twice. In the latter case, since $\E[X_L] = \E[X_G]$, condition (3.A.57) of \cite[Theorem~3.A.44]{shaked2007stochastic} yields
$\E[\varphi(X_L)] \le \E[\varphi(X_G)]$ for every convex function $\varphi$. $\blacksquare$
\end{appendix}
\end{document}